\documentclass[letterpaper,11pt]{article}

\usepackage{fullpage}

\usepackage{amsmath}
\usepackage{amsfonts}
\usepackage{amssymb}
\usepackage{amsthm}
\usepackage{amsbsy}
\usepackage{graphicx}
\usepackage{caption}
\usepackage{subcaption}
\usepackage{verbatim}
\usepackage{url}
\usepackage{bbm}
\usepackage{multirow}
\usepackage{bm}

\usepackage{tikz}
\usetikzlibrary{shapes,snakes}
\usetikzlibrary{decorations}

\usepackage[ruled]{algorithm2e}

\SetArgSty{textrm}  

\usepackage{ctable} 

\newtheorem{theorem}{Theorem}
\newtheorem{lemma}[theorem]{Lemma}

\newtheorem{cor}[theorem]{Corollary}

\theoremstyle{definition}
\newtheorem{example}{Example}

\newcommand{\bigO}{\mathcal{O}}

\newcommand{\s}{\mathbf{s}}
\newcommand{\tbold}{\mathbf{t}}
\newcommand{\kbold}{\mathbf{k}}

\newcommand{\score}{\text{\texttt{sc}}}
\newcommand{\val}{\text{\texttt{val}}}
\newcommand{\g}{\text{\texttt{g}}}
\newcommand{\poly}{\text{poly}}
\newcommand{\one}[1]{\mathbbm{1}\left\{#1\right\}}

\usepackage{authblk}

\begin{document}

\title{\bf Optimizing positional scoring rules \\ for rank aggregation\thanks{This work was partially supported by Caratheodory research grant E.114 from the University of Patras, by a PhD scholarship from the Onassis Foundation, and by the European Research Council (ERC) under grant number 639945 (ACCORD).}}

\author{Ioannis Caragiannis \quad Xenophon Chatzigeorgiou \quad George A. Krimpas}
\affil{\em Department of Computer Engineering and Informatics, University of Patras, Greece}

\author{Alexandros A. Voudouris}
\affil{\em Department of Computer Science, University of Oxford, UK}

\date{}

\maketitle

\begin{abstract}
Nowadays, several crowdsourcing projects exploit social choice methods for computing an aggregate ranking of alternatives given individual rankings provided by workers. Motivated by such systems, we consider a setting where each worker is asked to rank a fixed (small) number of alternatives and, then, a positional scoring rule is used to compute the aggregate ranking. Among the apparently infinite such rules, what is the best one to use? To answer this question, we assume that we have partial access to an underlying true ranking. Then, the important optimization problem to be solved is to compute the positional scoring rule whose outcome, when applied to the profile of individual rankings, is as close as possible to the part of the underlying true ranking we know. We study this fundamental problem from a theoretical viewpoint and present positive and negative complexity results and, furthermore, complement our theoretical findings with experiments on real-world and synthetic data.
\end{abstract}

\section{Introduction}
Social choice theory \cite{BC+16} studies voting rules (also known as social choice or social welfare functions) that compute a winning alternative or a ranking of the available alternatives from voter preferences. Typically, the preference of each voter is supposed to be a ranking over {\em all} available alternatives. We deviate from this assumption and, instead, we focus our attention to settings in which each voter (or, better, agent for our purposes) ranks only a small subset of the alternatives. Such {\em incomplete} rankings seem to be non-standard in the literature; the papers \cite{WGS16,DKNS01,S07} are some notable exceptions.

Our adoption of incomplete rankings is motivated by {\em crowdsourcing}~\cite{LA11} and {\em rating applications}. For example, assume that a requester would like to rank a huge set of alternatives using expert opinions from a crowd of workers. Asking each worker for her opinion on the whole set of alternatives (i.e., for a full ranking of them) could be problematic since, most probably, the worker will not be aware of most of the alternatives. Even if she tries to obtain additional information, coming up with consistent comparisons between alternatives that she knows well and alternatives that she has no idea about, would be rather impossible, given their huge number. Instead, this task would be much easier if workers focused on small fixed sets of alternatives. The requester could give each worker a different set of few alternatives to rank. Then, processing smaller inputs and merging them to come up with a global ranking of all alternatives would be easier for the requester as well. 

This approach has been recently exploited in the context of ordinal peer grading in massive open online courses (MOOCs); see the papers \cite{ALMRW16,CKV15,CKV16,RJ14,SBP+13,SW17} for approaches of this flavor. In such settings, the task of grading an exam with many participating students is outsourced to the students themselves. Each student is given a small number of exam papers of other students to rank, and the final grading (a ranking of all students) is obtained by aggregating the inputs provided by the students. Ordinal peer grading has also been used in the evaluation of proposals for research funding, e.g., by the Sensors and Sensing Systems (SSS) program of NSF in 2013 \cite{H13}, using a Borda-like method proposed earlier by Merrifield and Saari~\cite{MS09} (see also~\cite{KLMP15}).

An example of the rating application that we envision is as follows. Users of a hotel booking system are asked to rank hotels in a specific city, in which they have recently stayed. The goal of the system is to compute a full ranking of the hotels (or, possibly, different rankings depending on different relevant criteria, such as price, cleanliness, location, etc.) that can help new users. Clearly, each user can provide meaningful feedback for just a few hotels. Again, in this scenario, the system might ask each user to focus only on a subset of the hotels she knows. Similar examples of rating applications include systems related to ranking restaurants, universities, and so on.

Besides the different sets of alternatives each individual is asked to rank in the above scenarios, another implicit feature is that there is an {\em underlying true ranking} of all alternatives (e.g., the ranking of exam papers in terms of their quality or the ranking of hotels in terms of their facilities) that we would like to compute when aggregating the individual preferences. Can we do so using simple voting-like rules? We follow an {\em optimization approach} which can be described with the following question: 
\begin{quote}
Assuming that we have partial knowledge of the underlying true ranking and access to sampled profiles, which is the rule that yields an outcome that is as consistent as possible to (our partial knowledge of) the underlying true ranking when applied to the sampled profiles? 
\end{quote}

We study the above question for {\em positional scoring rules} (or, simply, scoring rules), which have played a central role in social choice theory. Two factors that have led to this decision are their simplicity and effectiveness; simplicity follows by their definition and effectiveness is justified by our experimental results. In particular, we consider settings in which each agent is asked to rank the same number $d$ of alternatives; this is consistent to the ordinal peer grading approach as it has been applied in MOOCs \cite{CKV16,RJ14} or in the NSF pilot \cite{H13}. A positional scoring rule in our setting is defined by a scoring vector $(s_1, s_2, ..., s_d)$. It takes as input the incomplete individual rankings of the agents and computes scores for alternatives as follows. An alternative gets $s_k$ points each time it is ranked $k$-th by an agent and its score is its total number of points. The final ranking is obtained by ordering all alternatives in terms of their scores, in non-increasing order. 

The input of our problem consists of a profile of individual incomplete rankings and desired relations for pairs of alternatives (to be thought of as parts of the underlying true ranking) with corresponding weights (indicating the importance of each relation). Given this input, we would like to compute the positional scoring rule, whose outcome ---when applied on the profile--- maximizes the total weight of the desired pairwise relations it satisfies. We refer to this seemingly fundamental optimization problem as {\sf OptPSR}, standing for ``Optimizing Positional Scoring Rules''. 

\subsection{Our contribution}
Our technical contribution consists of theoretical and experimental results for {\sf OptPSR}. We begin by presenting an exact algorithm, which we call \texttt{Regions}, that solves {\sf OptPSR} in time that depends exponentially only on the parameter $d$. Hence, \texttt{Regions} runs in polynomial time when $d$ is constant. For instances with high values of $d$, we have two approximation algorithms. The first one, which we call \texttt{BestApproval}, searches among the class of $t$-approval voting rules (that use scoring vectors with $t$ $1$s followed by $d-t$ $0$s, with $t \in [d]$) and returns the one that satisfies constraints of highest total weight. The solutions returned by \texttt{BestApproval} are always at least $1/d$--approximate. This means that the total weight of the satisfied constraints is at least $1/d$ times the maximum total weight of constraints that can be simultaneously satisfied by any scoring rule. We show that our analysis is tight by constructing simple instances, in which any approval voting rule is (at most) $1/d$-approximate. We also present a second, more sophisticated, approximation algorithm, called \texttt{ApxPSR}, which achieves even better approximation ratios at the expense of higher (but still polynomial) running time. On the negative side, we show that {\sf OptPSR} is not only computationally hard (in particular, NP-hard) to compute exactly but also NP-hard to approximate. We present an explicit inapproximability bound of $23/24$; our proof is based on an approximation-preserving reduction from the optimization problem MAX-3LIN-2 of maximizing the number of satisfied equations in an over-determined system of linear equations modulo $2$, and exploits a famous inapproximability result due to H{\aa}stad \cite{H01}. 

Further, we describe experiments from the execution of scoring rules/algorithms on many {\sf OptPSR} instances. We use two real-world profiles, which we have carefully collected, as well as numerous synthetic profiles that are produced by simulating agents whose ranking behavior follows the Bradley-Terry \cite{BT52} and Plackett-Luce~\cite{L59,P75} noise models. In contrast to our theoretical work, which is based on worst-case assumptions, our experimental results show that well-known scoring rules as well as our algorithms \texttt{ApxPSR} and \texttt{BestApproval} perform remarkably well and recover almost $100\%$ of the desired constraints in all scenarios we examined; this justifies our choice to study the optimization problem {\sf OptPSR} in the first place. 

\subsection{Related work}
Social choice theory has traditionally assumed that voters provide full rankings (strict linear orders) over all alternatives. There are some deviations from this tradition that have been attempted recently. Among other issues, Boutilier and Rosenschein~\cite{BR16} discuss settings with incomplete rankings as votes. In general, these models belong to one of the following categories. Several papers (see, e.g., \cite{BFL12,DLC15,FO14}) consider voters who rank their few top preferred alternatives. Others, like the current paper, consider voters who rank arbitrary subsets of alternatives. These include the papers \cite{WGS16,DKNS01,S07}, in which voters rank alternatives they know. For example, in \cite{DKNS01}, the alternatives are web pages and the voters are different search engines, which do not necessarily have full coverage of the web. In the papers on ordinal peer grading mentioned above \cite{ALMRW16,CKV15,CKV16,KLMP15, MS09,RJ14,SBP+13,SW17}, the ``voters'' are asked to rank particular subsets of alternatives. In this sense, even though they are in principle capable of evaluating other alternatives, they are never asked to do so and, hence, do not include them in their rankings. In the most general setting, each vote can be a {\em partial order} of all alternatives. This includes the cases mentioned above as well as the case in which some relations between pairs of alternatives are not given. Konczak and Lang~\cite{KL05} were the first to consider this setting in computational social choice, followed by Pini et al.~\cite{PRVW11}, Xia and Conitzer~\cite{XC11a}, and others. 

Important problems in the setting with partial orders of alternatives as votes are related to how the votes can be extended to full rankings. The papers~\cite{KL05,PRVW11,XC11a} mentioned above consider the question of whether there exists an extension of a profile of partial orders so that a given alternative becomes the winner. The question of whether a given alternative is the winner in {\em any} profile extension has also been studied. Both questions give rise to elegant computational problems, which are informally known as {\em possible} and {\em necessary winners}, respectively. Lu and Boutilier \cite{LB11a} view the selection among different possible winners as a robust optimization problem and use the notion of minimax regret to solve it. Their techniques extend to multi-winner determination~\cite{LB13}. However, all these papers focus on the winning alternative(s), while our interest here is on the whole ranking returned by a voting rule that is applied on profiles with incomplete votes.

The existence of an underlying true ranking is a central assumption in our work. This is closely related to a trend in social choice, which assumes an objectively correct ranking of the alternatives (a {\em ground truth}) and views votes as noisy estimates of this ranking. The most common approach in such settings aims to view voting rules as {\em maximum likelihood estimators} \cite{CS05,Y88} and start with the assumption that each voter implicitly transforms the ground truth into her vote following a particular probability distribution or noise model. A voting rule is an MLE for a noise model if, when applied on a profile of votes, it returns as an outcome the ranking or the winning alternative that is the most likely to produce this profile, assuming that voters follow the noise model. The most prominent such result is due to Young \cite{Y88}, who proved that the Kemeny voting rule is the MLE for a noise model that dates back to Marquis de Condorcet \cite{C85} and is today better known as the Mallows' model \cite{M57}. A discussion on recent results on the MLE approach can be found in the chapter by Elkind and Slinko \cite{ES16}. Among them, Xia and Conitzer \cite{XC11} and Lu and Boutilier~\cite{LB11} use the MLE approach to voting rules applied on profiles with incomplete votes. A related line of research aims to establish sample complexity results. How many votes (from a given noise model) are necessary in order to recover the ground truth with high probability? The papers \cite{BM08,CPS16,CK14} follow this direction. 

Another approach, which is even closer to the current paper, has an optimization flavor. The papers \cite{CKV15,CKV16} of our group on ordinal peer grading as well as the paper by de Weerdt et al. \cite{WGS16} aim to identify the voting rule whose outcome ranking has as small expected distance from the ground truth ranking as possible. In general, these voting rules are not maximum likelihood estimators. In contrast to these papers as well as to those following the MLE approach, here we assume that we have access to parts of the underlying ground truth. Furthermore, in our theoretical investigations, we do not exploit the fact that votes may be noisy estimates of some ground truth but, instead, treat them as arbitrary; this gives rise to several optimization challenges. On the other hand, our experimental scenarios use ground truth rankings and voters (agents) that follow two well-known noise models.

Finally, our optimization problem {\sf OptPSR} aims to compute the positional scoring rule that best fits the input. This is conceptually related to learning-theoretic studies that seek a scoring rule that is consistent to given examples; e.g., see the papers by Boutilier et al. \cite{BCH+15} and Procaccia et al. \cite{PZ+09} which, among other results, study the sample complexity of scoring rules by bounding their generalized dimension. Like our theoretical investigations, their settings are more general and do not depend on particular noise models.

\subsection{Roadmap}
The rest of the paper is structured as follows. We begin with the formal description of the {\sf OptPSR} problem and necessary preliminary material in Section~\ref{sec:problem}, including definitions, an example, as well as the description of a naive exact algorithm for {\sf OptPSR}. In Section~\ref{sec:constant}, we present and analyze our exact algorithm \texttt{Regions}. Our approximation algorithms \texttt{BestApproval} and \texttt{ApxPSR} are presented and analyzed in Section~\ref{sec:apx-approval}. In Section~\ref{sec:hard}, we give the proof of our inapproximability result for {\sf OptPSR}. Our experiments follow in Section~\ref{sec:exp}. We conclude with a discussion on open problems and possible extensions in Section \ref{sec:open}. Additional material related to our experiments is given in Appendix.

\section{Problem statement}\label{sec:problem}
We consider settings with a set $N$ of {\em agents} and a set $A$ of {\em alternatives}. Agent $i$ expresses her preference over a subset $A_i\subseteq A$ of alternatives; her preference (or vote) is a strict linear order (henceforth, simply, a {\em ranking}) of the alternatives in $A_i$. A preference profile (or simply, a {\em profile}) $\Pi$ consists of the preferences of all agents. In this work, we assume that all agents have the {\em same} number $d\geq 2$ of alternatives in their preferences, i.e., $|A_i|=d$ for each agent~$i$.

A social welfare function takes as input a profile $\Pi$ and it outputs a ranking of all alternatives in $A$. A {\em positional scoring rule} (or, simply, a {\em scoring rule}) is a social welfare function that uses a scoring vector $\s = (s_1, ..., s_d)$ with $s_i\geq s_{i+1}$ for $i=1, ..., d-1$ and $s_d\geq 0$; the alternative at position $k$ in each vote is assigned $s_k$ points and the ranking of the alternatives is produced by ordering them in monotone non-increasing order in terms of their total points (or {\em score}). 
In the following, with some abuse in notation, we use $\s$ to refer to both the scoring vector and the corresponding scoring rule that uses it.
Formally, for an alternative $x$, let $\nu_j(x,\Pi)$ denote the number of agents that rank $x$ at position $j$ in profile $\Pi$. Then, given a scoring rule $\s$, the score of alternative $x$ is defined as 
$$\score_{\s}(x,\Pi) = \sum_{j=1}^d{\nu_j(x,\Pi) \cdot s_j}.$$
We remark that this score definition does not take into account the popularity of alternative $x$ (i.e., the number of times $x$ appears in the rankings of the profile). Another definition would normalize the score by dividing with the number of appearances of $x$ in the profile. We have chosen the current definition for proof-of-concept purposes only.

We also assume that we have access to a set of {\em constraints} $C$ that represents our (possibly partial) knowledge to an objective set of pairwise relations between the alternatives. Each constraint in $C$ is given by an ordered pair of alternatives $(x,y)$ and requires that alternative $x$ is ranked higher than alternative $y$ in the outcome of the scoring rule $\s$. For a pair of alternatives $(x,y)$, let $\delta_j(x,y,\Pi) = \nu_j(x,\Pi) - \nu_j(y,\Pi)$. Now, observe that, in order for alternative $x$ to be ranked above $y$ with certainty in the final ranking, it must be $\score_\s(x,\Pi) > \score_\s(y,\Pi)$ and, equivalently,
$$\sum_{j=1}^d {\delta_j(x,y,\Pi) \cdot s_j}> 0.$$
Using $\boldsymbol{\delta}(x,y,\Pi) = (\delta_1(x,y,\Pi), ..., \delta_d(x,y,\Pi))$, the above expression can be compactly written as the dot product $\boldsymbol{\delta}(x,y,\Pi) \cdot \s>0$. 
For our purposes, instead of thinking of a profile $\Pi$ as the set of rankings provided by the agents, it is convenient to describe it using the quantities $\boldsymbol{\delta}(x,y,\Pi)$ for every constraint $(x,y)$ in $C$; we use the notation $\boldsymbol{\delta}(\Pi)$ to denote the set of these quantities and we will simply refer to it as the profile. Each constraint $(x,y)\in C$ has a corresponding non-negative weight $w(x,y)$, which indicates the importance of the constraint. 

Now, problem {\sf OptPSR} (standing for ``optimizing positional scoring rules'') is defined as follows. We are given a profile $\boldsymbol{\delta}(\Pi)$ and a set $C$ of constraints. The goal of {\sf OptPSR} is to find the scoring rule $\s$ that produces a ranking of all alternatives so that the total weight (or gain)
$$\g(\s,\boldsymbol{\delta}(\Pi),C) = \sum_{(x,y)\in C}w(x,y) \cdot \one{\boldsymbol{\delta}(x,y,\Pi) \cdot \s >0},$$
of satisfied constraints is maximized. The quantity $\one{X}$ takes value $1$ if $X$ is true and $0$ otherwise. 

\begin{example}\label{example}
Consider ten agents, a set of seven alternatives $A = \{x_1,x_2,x_3,x_4,x_5,x_6,x_7\}$, the profile of rankings of size $d=4$ that appears in Table~\ref{example:profile}, and the constraints that appear in Table~\ref{example:delta-profile} together with the corresponding weights and the representation $\boldsymbol{\delta}(\Pi)$ of the profile. 

First, observe that the first three constraints cannot be satisfied simultaneously; this can be easily seen since by summing the corresponding inequalities we obtain $-s_2 + s_3 - 8s_4 > 0$ which contradicts $s_2 \geq s_3$ and $s_4 \geq 0$. So, in the best case, the optimal scoring rule can satisfy all constraints besides the one among the first three that has the minimum weight. One such scoring rule uses the scoring vector $(4,4,1,0)$ and satisfies all constraints except the second one for a total gain of $10$. In contrast, well-known scoring rules such as the Borda count that uses the scoring vector $(3,2,1,0)$ and the $t$-approval rules, with $t \in [4]$, that use scoring vectors with $t$ ones followed by $4-t$ zeros, satisfy constraints of total gain equal to $7$, $8$ (for $t=1$), $8$ (for $t=2$), $9$ (for $t=3$), and $4$ (for $t=4$), respectively; see also Table~\ref{example:rules-gains}. \hfill$\qed$

\begin{table}[h!]
\centering
\begin{tabular}{c c}
\noalign{\hrule height 1pt}\hline
\# of agents & ranking 	\\\hline
$1$ & $x_1 \succ x_2 \succ x_6 \succ x_4$\\
$2$ & $x_3 \succ x_1 \succ x_4 \succ x_2$\\
$2$ & $x_5 \succ x_6 \succ x_4 \succ x_2$\\
$1$ & $x_7 \succ x_3 \succ x_4 \succ x_2$\\
$1$ & $x_5 \succ x_4 \succ x_3 \succ x_6$\\
$3$ & $x_7 \succ x_5 \succ x_4 \succ x_2$\\
\noalign{\hrule height 1pt}\hline
\end{tabular}\normalsize
\caption{The profile $\Pi$ of agent rankings in Example~\ref{example}. The notation $\succ$ is used here to represent the preference of the agents. For instance, according to the second row of the table, there are two agents that rank alternative $x_3$ first, alternative $x_1$ second, $x_4$ third, and $x_2$ last.}
\label{example:profile}
\end{table} 

\begin{table}[h!]
\centering
\begin{tabular}{c c c c c c c}
\noalign{\hrule height 1pt}\hline
constraint & weight & $\delta_1(\cdot,\Pi)$ & $\delta_2(\cdot,\Pi)$ & $\delta_3(\cdot,\Pi)$ & $\delta_4(\cdot,\Pi)$ & inequality \\\hline
$(x_1,x_2)$ & $4$ & $1$ & $1$ & $0$ & $-8$ & $s_1 + s_2 - 8s_4 > 0$ \\
$(x_4,x_5)$ & $2$ & $-3$ & $-2$ & $8$ & $1$ & $-3s_1 -2s_2 + 8s_3 + s_4 > 0$\\
$(x_3,x_4)$ & $3$ & $2$ & $0$ & $-7$ & $-1$ & $2s_1 -7s_3 -s_4 > 0$\\
$(x_4,x_6)$ & $2$ & $0$ & $-1$ & $7$ & $0$ & $-s_2 + 7s_3 > 0$\\
$(x_3,x_2)$ & $1$ & $2$ & $0$ & $1$ & $-8$ & $2s_1 + s_3 -8s_4 > 0$\\
\noalign{\hrule height 1pt}\hline
\end{tabular}\normalsize
\caption{The constraints, the corresponding weights, the alternative representation of the profile $\Pi$ using the quantities $\boldsymbol{\delta}(x,y,\Pi)$, and the induced inequalities used in Example~\ref{example}. For instance, the first constraint of weight $4$ requires that alternative $x_1$ appears above $x_2$ in the final ranking. According to the profile $\Pi$ in Table~\ref{example:profile}, since there are two agents that place alternative $x_1$ in the second position, while there is only one agent that places alternative $x_2$ in the second position, we have that $\delta_2(x_1,x_2,\Pi)=1$; one can easily verify the remaining values of the table. Given these quantities, the inequalities follow since $\delta_j(x_1,x_2,\Pi)$ for $j \in [4]$ is the coefficient of the variable $s_j$ corresponding to the points assigned to position $j$.}
\label{example:delta-profile}
\end{table} 
\end{example}

\begin{table}[h!]
\centering
\begin{tabular}{c c c}
\noalign{\hrule height 1pt}\hline
rule & scoring vector & gain 	\\\hline
opt & $(4,4,1,0)$ & $10$ \\
Borda count & $(3,2,1,0)$ & $7$ \\
$1$-approval & $(1,0,0,0)$ & $8$ \\
$2$-approval & $(1,1,0,0)$ & $8$ \\
$3$-approval & $(1,1,1,0)$ & $9$ \\
$4$-approval & $(1,1,1,1)$ & $4$ \\
\noalign{\hrule height 1pt}\hline
\end{tabular}\normalsize
\caption{The positional scoring rules considered in Example~\ref{example} and their corresponding total gains.}
\label{example:rules-gains}
\end{table}

Let us now give an equivalent view of {\sf OptPSR}. A scoring rule $\s$ can be thought of as a point in $\mathbb{R}^d$, and, in particular, in the region $R_0$ of $\mathbb{R}^d$ formed by the inequalities $s_i-s_{i+1}\geq 0$ for $i=1, ..., d-1$ and $s_d\geq 0$ that define all valid scoring vectors. We can define subregions of $R_0$ by considering any subset $C'\subseteq C$ of constraints and the inequality $\boldsymbol{\delta}(x,y,\Pi)\cdot \s >0$ for every constraint associated with the pair of alternatives $(x,y)\in C'$ and the inequality $\boldsymbol{\delta}(x,y,\Pi)\cdot \s \leq 0$ for every constraint $(x,y)\in C\setminus C'$. In this way, the collection of all subsets of constraints in $C$ partition $R_0$ into disjoint subregions; of course, some of them may be infeasible. Hence, in order to maximize $\g(\s,\boldsymbol{\delta}(\Pi),C)$, it suffices to find any point $\s$ in the non-empty subregion of $R_0$ that satisfies the subset of constraints with maximum total weight. 

To do so, we can enumerate all subsets of constraints of $C$, check feasibility of the corresponding regions using linear programming, and report any point in the subregion that yields the highest gain. This algorithm takes time polynomial in 
$2^{|C|}$ and $d$, assuming that it receives $\boldsymbol{\delta}(\Pi)$ and $C$ as input. In practice, the parameter $d$ (i.e., the size of the input rankings) is expected to be a small constant, while the number of alternatives and, consequently, the number of constraints $|C|$ would be much larger. Hence, an algorithm that runs in time exponential in $|C|$ is clearly impractical.
In the next section, we will present an algorithm that uses a more clever enumeration of the feasible subregions in order to get the one that yields the maximum gain.

\section{An improved {\sf OptPSR} algorithm}\label{sec:constant}

We will present another (exact) {\sf OptPSR} algorithm which we call \texttt{Regions}. Its running time depends exponentially only on the parameter $d$ and, hence, is polynomial when $d$ is a constant. We remark that this time complexity is interesting only in theory. As we will mention later, when describing our experiments in Section~\ref{sec:exp}, even when $d$ is small (say, equal to $6$), the algorithm does not scale well with the number of constraints.

\texttt{Regions} computes a pool of non-empty subregions of $R_0$, each of which satisfies a different subset of constraints. Initially, the pool consists only of region $R_0$, and is updated as new constraints of $C$ are considered. When a new constraint is considered, each region in the pool can be split into two subregions consisting of the points that satisfy the constraint and the points that do not satisfy it, respectively. When all points of a region satisfy the constraint or all points of the region do not satisfy it, then the region is not split and is retained as a whole in the pool.

In particular, the algorithm considers the constraints of $C$ one by one. At each step $t$, a pool ${\cal P}$ of regions is kept; at the beginning of each step, all regions in the pool are active. For each region $R$ in ${\cal P}$, the algorithm keeps the gain $\val(R)$ that is obtained by the constraints which have been considered until step $t$ and are satisfied by scoring vectors of region $R$. The algorithm begins its execution having only region $R_0$ in the pool. When a new constraint $(x,y)$ with weight $w(x,y)$ is considered, the algorithm attempts to update each active region $R$ of ${\cal P}$ as follows. It defines the candidate regions $R^{xy}$ and $R^{\neg xy}$ such that
\begin{itemize}
\item $R^{xy}$ is defined by the inequalities that form $R$ together with inequality $\boldsymbol{\delta}(x,y,\Pi)\cdot \s >0$ (which defines the set of points that satisfy constraint $(x,y)$), and
\item $R^{\neg xy}$ is defined by the inequalities that form $R$ together with inequality $\boldsymbol{\delta}(x,y,\Pi)\cdot \s \leq 0$ (which defines the set of points that do not satisfy constraint $(x,y)$).
\end{itemize}
If both $R^{xy}$ and $R^{\neg xy}$ are non-empty (i.e., the corresponding sets of inequalities are feasible), the algorithm includes both $R^{xy}$ and $R^{\neg xy}$ in ${\cal P}$ as inactive, sets their gains $\val(R^{xy}):=\val(R)+w(x,y)$ and $\val(R^{\neg xy}):=\val(R)$, and removes region $R$ from the pool. If only $R^{xy}$ is feasible (and $R^{\neg xy}$ is infeasible), $\val(R)$ is increased by $w(x,y)$. If only $R^{\neg xy}$ is feasible (and $R^{xy}$ is infeasible), the algorithm does nothing. In the last two cases, no new region is added to the pool. Clearly, it cannot be the case that both $R^{xy}$ and $R^{\neg xy}$ are infeasible. Again, feasibility can be checked efficiently by solving linear programs with $d$ variables and up to $|C|$ constraints. At the end of step $t$ (i.e., when there is no other active region in the pool to be considered), the inactive regions become active and the algorithm proceeds with step $t+1$. When all constraints of $C$ have been considered, the algorithm computes the active region $R^*$ with maximum $\val(R^*)$ and returns any scoring vector in $R^*$. A description of the algorithm in pseudocode is given as Algorithm~\ref{alg:regions}.

\begin{algorithm}[h!]
    \SetAlgoLined
    \KwIn{parameter $d$, profile $\boldsymbol{\delta}(\Pi)$, set $C$ of constraints}
    \KwOut{a scoring vector $\s = (s_1, ..., s_d)$}
	$R_0 := \{s_1 \geq s_2, ..., s_{d-1}\geq s_d, s_d \geq 0\}$   \\
	$\val(R_0) := 0$ \\
    ${\cal P} := \{R_0\}$ \\
	\For{$(x,y) \in C$}{
		$\text{active} := {\cal P}$ \\
		\For{$R \in \text{active}$}{
			$\text{active} := \text{active}\setminus \{R\}$	\\
			$R^{xy} := \{R, \boldsymbol{\delta}(x,y,\Pi)\cdot \s > 0 \}$ \\
			$R^{\neg xy} := \{R, \boldsymbol{\delta}(x,y,\Pi)\cdot \s \leq 0 \}$ \\
			\uIf{$R^{xy}$ is feasible and $R^{\neg xy}$ is feasible}{
				$\val(R^{xy}) := \val(R) + w(x,y)$ \\
				$\val(R^{\neg xy}) := \val(R)$ \\
				${\cal P} := \{{\cal P} \setminus \{R\} \cup\{ R^{xy}, R^{\neg xy}\} \}$ 
			}
			\ElseIf{$R^{xy}$ is feasible and $R^{\neg xy}$ is not feasible}{
				$\val(R) := \val(R) + w(x,y)$
			}
		}	
	}    
	$R^* := \arg\max_{R \in {\cal P}} \{ \val(R) \}$ \\
	return $\s \in R^*$
\caption{\texttt{Regions}}
\label{alg:regions}
\end{algorithm}
 
\begin{example}
We will now examine a simple example of how \texttt{Regions} works on a profile $\Pi$ with alternatives $x_1$, $x_2$, $x_3$, $y_1$, $y_2$, and $y_3$, and $d=2$ (see Figure~\ref{fig:alg}). The set $C$ has three constraints $(x_1,y_1)$, $(x_2,y_2)$, and $(x_3,y_3)$ with corresponding weights $3$, $1$ and $2$. The profile is such that $\boldsymbol{\delta}(x_1,y_1,\Pi)=(-7,2)$, $\boldsymbol{\delta}(x_2,y_2,\Pi)=(4,-2)$, and $\boldsymbol{\delta}(x_3,y_3,\Pi)=(-2,3)$. Therefore, the constraints define the inequalities $-7s_1+2s_2>0$, $4s_1-2s_2>0$, and $-2s_1+3s_2>0$. 

Now, the algorithm proceeds as follows. Initially (see Figure~\ref{fig:alg}a), the algorithm has region $R_0$, defined by the lines $s_2=0$ and $s_1 - s_2 = 0$, in the pool with gain equal to $0$. At the next step (see Figure~\ref{fig:alg}b), the algorithm considers constraint $(x_1,y_1)$ and replaces $R_0$ with regions $R_1=R_0^{x_1y_1}$ and $R_2=R_0^{\neg x_1y_1}$; $R_1$ is the subregion of $R_0$ with $-7s_1+2s_2>0$ that satisfies the first constraint and has gain $3$, while $R_2$ is the subregion of $R_0$ with $-7s_1+2s_2<0$ that does not satisfy the first constraint and has gain $0$ (the line $-7s_1+2s_2=0$ separates the two subregions). Next (see Figure~\ref{fig:alg}c), the constraint $(x_2,y_2)$ leaves both regions $R_1$ and $R_2$ in the pool, and increases both of their gains by $1$. Finally (see Figure~\ref{fig:alg}d), the third constraint $(x_3,y_3)$ replaces region $R_1$ by regions $R_3 = R_1^{x_3y_3}$ and $R_4 = R_1^{\neg x_3y_3}$; $R_3$ is the subregion of $R_1$ with $-2s_1+3s_2>0$ that satisfies the constraint and has gain equal to $6$ (the gain of $R_1$ increased by the weight of the constraint), while $R_4$ is the subregion of $R_1$ with $-2s_1+3s_2<0$ that does not satisfy the constraint and has gain $4$ (equal to that of $R_1$). Observe that in the last step, region $R_2$ also satisfies the third constraint and, hence, its gain is also increased by $2$. The region with the maximum gain is $R_3$ and the algorithm will output some scoring vector from this region. \hfill$\qed$
\end{example}

\begin{figure}[p]
\centering
\includegraphics[scale=0.85]{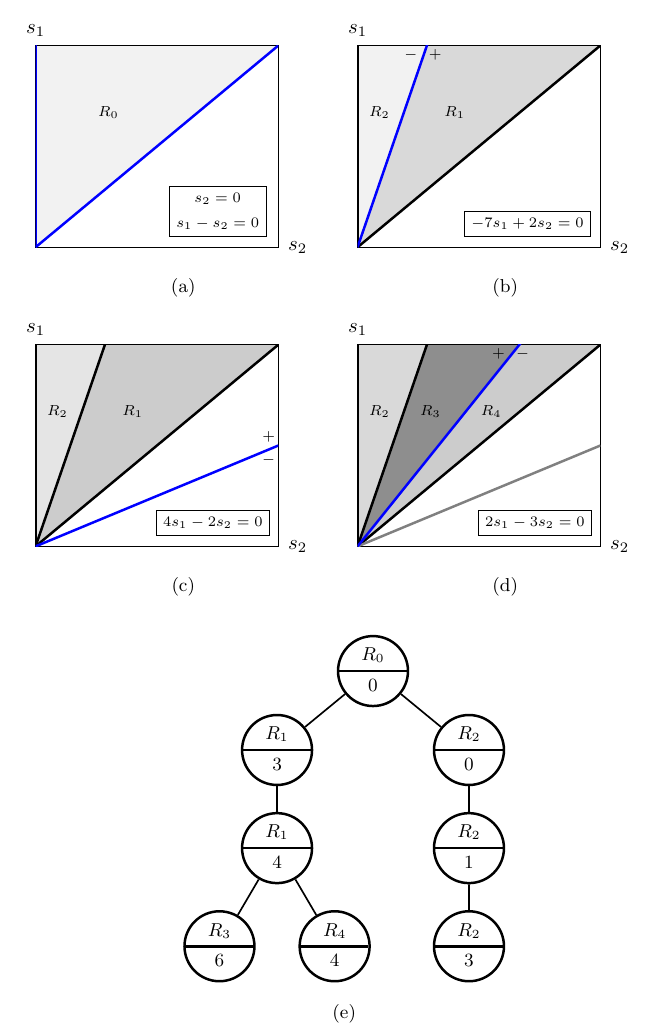}
\caption{An example with the execution of \texttt{Regions} on a profile $\Pi$ with $d=2$. Subfigures (a)--(d) depict the step-by-step consideration of the constraints and how the active regions are updated. At each subfigure, the blue line corresponds to a new constraint that is considered. In Subfigures (b)--(d), the marks $+$ and $-$ denote which of the two subregions defined by the blue line contain the vectors that satisfy the constraint or not. 
The gain of the subregions marked with $+$ are increased by the weight of the constraint (a darker shade of gray for a subregion indicates a higher gain). Subfigure (e) depicts the evolution of the content of the pool of regions together with the corresponding gains. Specifically, the nodes in each level of the tree represent the content of the pool and the corresponding gains during each step of the algorithm.}
\label{fig:alg}
\end{figure}

Next, we prove that our improved {\sf OptPSR} algorithm is correct and that its running time depends exponentially only on the parameter $d$.

\begin{theorem}\label{thm:improved-alg}
Given an instance of {\sf OptPSR} with parameter $d$, a set of constraints $C$, and a profile $\Pi$, algorithm \texttt{Regions} correctly returns a solution in time {\em $\bigO(|C|^d\cdot \poly(|C|,d))$}.
\end{theorem}

\begin{proof}
The correctness of the algorithm should be apparent. It considers the whole space of points in $\mathbb{R}^d$ 
which corresponds to scoring vectors and divides it into all (sub)regions defined for every inclusion-maximal subset of constraints that are satisfied simultaneously. Among all these regions, it finds the one with points that correspond to scoring vectors that satisfy constraints of $C$ with maximum total weight.

Expanding $R_0$ into the regions in the pool when the last constraint of $C$ is considered can be thought of as a non-complete binary tree $T$ with nodes corresponding to regions (see Figure \ref{fig:alg}e for an example). $T$ is rooted at a node corresponding to $R_0$ and is such that each node at level $t-1$, corresponding to a region $R$, has two children at level $t$ if the region $R$ was split in and replaced by two subregions at step $t$ and has one child otherwise (indicating that the region was retained in the pool during step $t$). The total time required to find all regions is proportional to the size of $T$. Since all non-leaf nodes have at least one child, the size of $T$ is at most its height $|C|$ times the number of leaves. The number of leaves is essentially the number of different non-empty regions, which is upper-bounded by the number of different sign patterns that the quantities $\boldsymbol{\delta}(x,y,\Pi) \cdot \s$ define for each constraint $(x,y)$ in $C$. Since these $|C|$ quantities are linear functions over the $d$ coordinates of vector $\s$, a result due to Alon \cite{A96} (see also Warren \cite{W68}) yields that the total number of different sign patterns is at most $\left(\frac{8e|C|}{d}\right)^d$. For each of the nodes of $T$, feasibility can be checked by solving two linear programs with $d$ variables and at most $|C|$ constraints in time $\poly(|C|,d)$. The theorem follows. 
\end{proof}

By Theorem \ref{thm:improved-alg}, we obtain the following corollary. For comparison, the naive algorithm presented at the end of the previous section is polynomial in the very special case where $|C|$ is at most logarithmic in $d$. 
\begin{cor}
Algorithm \texttt{Regions} solves instances of {\sf OptPSR} with constant $d$ in polynomial time.
\end{cor}

\section{Approximating {\sf OptPSR}}\label{sec:apx-approval}
As the running time of the exact algorithm \texttt{Regions} of the previous section depends exponentially on $d$, our aim here is to design much faster (i.e., polynomial-time) algorithms that compute approximate {\sf OptPSR} solutions. Given an instance of {\sf OptPSR} with parameter $d$, profile $\Pi$, and set $C$ of constraints, let $\s^*$ be the scoring vector that satisfies constraints of $C$ with maximum total weight. A scoring vector $\s$ is a $\rho$-approximate solution, for some $\rho \in [0,1]$, for the particular instance if $\g(\s,\boldsymbol{\delta}(\Pi),C) \geq \rho \cdot \g(\s^*,\boldsymbol{\delta}(\Pi),C)$, i.e., the total weight of constraints satisfied by $\s$ is at least $\rho$ times the total weight of constraints satisfied by $\s^*$. An algorithm is called a $\rho$-approximation algorithm if it computes a $\rho$-approximate solution for every instance of {\sf OptPSR}. We refer to $\rho$ as the {\em approximation ratio} of the algorithm. Ideally, we would like to design approximation algorithms that are as efficient as possible, i.e., algorithms that run in polynomial-time and have as high approximation ratio as possible.

Let us warm up by observing that a single positional scoring rule (e.g., Borda, plurality, $k$-approval) cannot serve as an efficient approximation algorithm as it has an approximation ratio of $0$. This is stated in the next lemma.

\begin{lemma}\label{lem:any-rule-counterexample}
Let $d$ be a positive integer. For every scoring vector $\s\in \mathbb{R}^d_{\geq 0}$ with $s_1 \geq ... \geq s_d$, there exists an instance of {\sf OptPSR} with parameter $d$, profile $\Pi$ and set $C$ of constraints such that $\s$ is $0$-approximate.
\end{lemma}

\begin{proof}
Clearly, the scoring vector $\s=(0, ..., 0)$ does not satisfy any constraint. So, in the following, we assume that $s_1>0$. For any positive integer $K>0$, we will construct a set $C$ of constraints consisting of disjoint pairs of alternatives $(x_t,y_t)$ for $t=1, ..., K$. We will distinguish between two cases depending on the structure of $\s$. For each of them we will construct an {\sf OptPSR} instance (formed by an appropriately defined profile and the set $C$ of constraints), in which $\s$ satisfies no constraint, while another scoring vector $\s^*$ satisfies all of them.

If $s_1=...=s_d>0$, then for every constraint $(x_t,y_t)$ we set $\delta_1(x_t,y_t,\Pi)=1$, $\delta_2(x_t,y_t,\Pi)=-2$ and $\delta_j(x_t,y_t,\Pi)=0$ for $j=3, ..., d$. Profile $\Pi$ can be realized as follows. Alternative $x_t$ appears in position $1$ once and alternative $y_t$ appears in position $2$ twice; all other positions are filled with additional alternatives that do not appear in the constraints of $C$. Observe that $\boldsymbol{\delta}(x_t,y_t,\Pi) \cdot \s =s_1-2s_2<0$ for every $t=1, ..., K$, i.e., $\s$ does not satisfy any constraint. In contrast, the plurality vector $\s^*=(1,0,...,0)$ satisfies all constraints of $C$.

Otherwise, if $s_i>s_{i+1}$ for some $i\in [d-1]$, let $D$ be an integer satisfying $D>\frac{s_{i+1}}{s_i-s_{i+1}}$. For every constraint $(x_t,y_t)$ for $t=1, ..., K$, we set $\delta_i(x_t,y_t,\Pi)=-D$, $\delta_{i+1}(x_t,y_t,\Pi)=D+1$, and $\delta_j(x_t,y_t,\Pi)=0$ for $j\in [d]\setminus \{i,i+1\}$. Profile $\Pi$ can be realized as follows. Alternative $x_t$ appears $D$ times in position $i$ and alternative $y_t$ appears $D+1$ times in position $i+1$; again, all other positions are filled with additional alternatives that do not appear in the constraints of $C$. Observe that the definition of $D$ implies that $\boldsymbol{\delta}(x_t,y_t,\Pi) \cdot \s =-s_iD+s_{i+1}(D+1)<0$, i.e., $\s$ does not satisfy any constraint. In contrast, the $(i+1)$-approval vector $\s^*$ (consisting of $1$s in the first $i+1$ positions and $0$s in the remaining ones) satisfies all constraints of $C$.

In both cases, the scoring vector $\s$ is a $0$-approximate solution.
\end{proof}

We conclude that efficient approximation algorithms should consider many candidate scoring vectors and pick the one that better serves a given {\sf OptPSR} instance. This is a recipe that is followed by the algorithms \texttt{BestApproval} and \texttt{ApxPSR} that are presented in Sections~\ref{subsec:bestapproval} and \ref{subsec:apxpsr}, respectively.

\subsection{Approximating {\sf OptPSR} using approval scoring vectors}\label{subsec:bestapproval}
We will now show that an extremely simple algorithm that examines a set of simple scoring vectors and returns the best of them achieves a $1/d$-approximation solution.  Formally, for $t\in [d]$, the $t$-approval rule is a positional scoring rule that uses the scoring vector that has $1$ in the first $t$ positions and $0$ in the remaining ones. Our algorithm, which we call \texttt{BestApproval}, considers all $t$-approval rules and returns the one that satisfies constraints of maximum total weight. A description of \texttt{BestApproval} in pseudocode is given as Algorithm~\ref{alg:bestapproval}.

\begin{algorithm}[h!]
    \SetAlgoLined
    \KwIn{parameter $d$, profile $\boldsymbol{\delta}(\Pi)$, set $C$ of constraints}
    \KwOut{a $t^*$-approval rule}
	\For{$t \in [d]$}{
		$\tbold := (\underbrace{1, ..., 1,}_{t \text{ times}} 0, ..., 0)$ \\
		$\g(\tbold,\boldsymbol{\delta}(\Pi),C) := \sum_{(x,y)\in C}w(x,y) \cdot \one{\boldsymbol{\delta}(x,y,\Pi) \cdot \tbold >0}$
	 }   
	 return $t^* \in \arg\max_{t} \{ \g(\tbold,\boldsymbol{\delta}(\Pi),C) \}$
\caption{\texttt{BestApproval}}
\label{alg:bestapproval}
\end{algorithm}

With the following two theorems, we prove that \texttt{BestApproval} is a $1/d$-approximation algorithm for {\sf OptPSR} (Theorem~\ref{thm:approval-upper}) and, furthermore, we show that this bound is tight by providing a particular instance for which any $t$-approval rule is an (at most) $1/d$-approximate solution (Theorem~\ref{thm:approval-lower}).

\begin{theorem}\label{thm:approval-upper}
Given an instance of {\sf OptPSR} with parameter $d$, a set of constraints $C$, and a profile $\Pi$, algorithm \texttt{BestApproval} returns a $1/d$-approximate solution. 
\end{theorem}

\begin{proof}
In order to prove the bound on the approximation ratio, we will show that there exists some $t \in [d]$ so that the corresponding $t$-approval scoring rule is a $1/d$-approximate solution. Then, the $t^*$-approval rule returned by \texttt{BestApproval} will have an at least as good approximation guarantee. 

Let $\s^* \in \mathbb{R}_{\geq 0}^d$ be the optimal scoring rule for the given instance of {\sf OptPSR}, and let $X \subseteq C$ be the set of constraints that it satisfies, i.e., $\boldsymbol{\delta}(x,y,\Pi)\cdot \s^* >0$ if and only if $(x,y)\in X$. For $k\in [d]$, let \begin{align*}
X_{k} &=\left\{(x,y)\in X: \sum_{j=1}^k{\delta_j(x,y,\Pi)} > 0\right\}.
\end{align*}
I.e., $X_k$ contains the constraints of $X$ which are satisfied by the $k$-approval scoring rule. This definition implies that
\begin{align}\label{eq:app}
\sum_{k=1}^d{\g(\kbold,\boldsymbol{\delta}(\Pi),C)} &\geq \sum_{k=1}^d{\sum_{(x,y)\in X_k}{w(x,y)}}.
\end{align}

We now claim that each constraint $(x,y)\in X$ belongs to $X_k$ for some $k\in [d]$ and, hence,
\begin{align}\label{eq:app-2}
\sum_{(x,y)\in X}{w(x,y)} & \leq \sum_{k=1}^d{\sum_{(x,y)\in X_k}{w(x,y)}}.
\end{align}
Assume otherwise; then, it would mean that $\sum_{j=1}^k{\delta_j(x,y,\Pi)}\leq 0$ for every $k\in [d]$. By multiplying these $d$ inequalities with the non-negative quantities $s^*_k-s^*_{k+1}$ for $k=1, ..., d-1$ and $s^*_d$, and summing them, we obtain that
\begin{align*}
\sum_{k=1}^{d-1}{(s^*_k-s^*_{k+1})\sum_{j=1}^k{\delta_j(x,y,\Pi)}}+s^*_d\sum_{j=1}^d{\delta_j(x,y,\Pi)} & \leq 0.
\end{align*}
The claim follows by observing that the left-hand side is equal to $\boldsymbol{\delta}(x,y,\Pi)\cdot \s^*$ and, hence, $\boldsymbol{\delta}(x,y,\Pi)\cdot \s^* \leq 0$, contradicting the fact that $(x,y)\in X$.

Using (\ref{eq:app-2}) and (\ref{eq:app}), we obtain
\begin{align*}
\g(\s^*,\boldsymbol{\delta}(\Pi),C) &= \sum_{(x,y)\in X}{w(x,y)}\leq \sum_{k=1}^d{\sum_{(x,y)\in X_k}{w(x,y)}} \leq \sum_{k=1}^d{\g(\kbold,\boldsymbol{\delta}(\Pi),C)}.
\end{align*} 
This implies that there exists $t\in [d]$ such that $\g(\tbold,\boldsymbol{\delta}(\Pi),C)\geq \frac{1}{d}\g(\s^*,\boldsymbol{\delta}(\Pi),C)$, as desired.
The theorem follows.
\end{proof}

\begin{theorem}\label{thm:approval-lower}
There exists an instance of {\sf OptPSR} with parameter $d$ for which all $t$-approval rules, with $t \in [d]$, are (at most) $1/d$-approximate.
\end{theorem}

\begin{proof}
We will define an {\sf OptPSR} instance with $d$ pairs of alternatives $(x_t, y_t)$ as constraints with $w(x_t,y_t)=1$ for $t\in [d]$.  We will build a profile $\Pi$ so that the $t$-approval scoring rule satisfies only constraint $(x_t,y_t)$ for $t \in [d]$, while there exists a scoring rule that simultaneously satisfies all constraints. 
The profile is defined as follows:
\begin{itemize}
\item Alternative $x_1$ appears $2d-1$ times in position $1$, and alternative $y_1$ appears $2d$ times in position $2$. This means that $\delta_1(x_1,y_1,\Pi)=2d-1$, $\delta_2(x_1,y_1,\Pi)=-2d$ and $\delta_j(x_1,y_1,\Pi)=0$ for $j\geq 3$.

\item For $2 \leq t \leq d-1$, alternative $x_t$ appears $2d$ times in position $t$, and alternative $y_t$ appears once in position $1$ and $2d$ times in position $t+1$. This means that $\delta_1(x_t,y_t,\Pi)=-1$, $\delta_t(x_t,y_t,\Pi)=2d$, $\delta_{t+1}(x_t,y_t,\Pi)=-2d$, and $\delta_j(x_t,y_t,\Pi)=0$ for $j\not\in\{1,t,t+1\}$.

\item Alternative $x_d$ appears $2d$ times in position $d$, and alternative $y_d$ appears once in position $1$. This means that $\delta_1(x_d,y_d,\Pi) = -1$, $\delta_d(x_d,y_d,\Pi)=2d$, and $\delta_j(x_d,y_d,\Pi)=0$ for $2 \leq j \leq d-1$.

\item The rest of the positions in the votes are filled with additional alternatives that do not appear in the constraints.
\end{itemize}

Observe that, for every $t\in [d]$, it holds that $\sum_{j=1}^{t}\delta_{j}(x_t,y_t,\Pi) = 2d-1 > 0$ and $\sum_{j=1}^{t}\delta_{j}(x_\ell,y_\ell,\Pi)=-1$, for any $\ell \neq [d]\setminus\{t\}$. Hence, the $t$-approval scoring rule satisfies only constraint $(x_t,y_t)$ for a total weight of $1$. Now, consider a variation of the Borda count scoring rule that uses the scoring vector $\s = (d, d-1, ..., 1)$. Then, since $\boldsymbol{\delta}(x_\ell,y_\ell,\Pi) \cdot \s = d > 0$ for every $\ell \in [d]$, this scoring rule satisfies all constraints for a total weight of $d$. Therefore, we conclude that any $t$-approval has approximation ratio at most $1/d$.
\end{proof}

\subsection{An improved approximation algorithm}\label{subsec:apxpsr}
Here, we will design the more sophisticated approximation algorithm \texttt{ApxPSR}$_k$, which is parameterized by a positive integer $k$, and exploits ideas that we have already presented above. Similarly to \texttt{BestApproval}, \texttt{ApxPSR}$_k$ searches over a set of candidate scoring vectors and identifies the one that better serves the available input data (profile and constraints). Two important differences between \texttt{ApxPSR}$_k$ and \texttt{BestApproval} are that (a) the set of candidate scoring rules is much broader now and (b) a few executions of a variation of the exact algorithm \texttt{Regions} are used in order to find the best among these candidates.

Let $\ell\in [\lceil d/k\rceil]$. We say that a scoring vector $\s$ follows the $\ell$-th $k$-pattern if $s_1 = ... = s_{k(\ell-1)+1} \geq s_{k(\ell-1)+2} \geq ... \geq s_{\min\{d,k\ell\}}\geq 0$ and, if $\ell<\lceil d/k \rceil$, $s_{k\ell+1} = ... = s_{d}=0$. For example, the $t$-approval scoring rule follows that $t$-th $1$-pattern. Note that the scoring vectors that follow some of the $\lceil d/k\rceil$ $k$-patterns have a very special structure and (at most) $k$ different values in their score entries. 

For a given instance of {\sf OptPSR} and integers $k>0$ and $\ell\in [\lceil d/k\rceil]$, we can compute the best scoring vector that follows the $\ell$-th $k$-pattern via a slight modification of \texttt{Regions}. We refer to this modification as algorithm \texttt{mRegions} and we assume that, together with parameter $d$, the profile $\boldsymbol{\delta}(\Pi)$, and the set $C$ of constraints, it receives as input the parameters $k$ and $\ell$ as well. All we need to do is to include the equality constraints which restrict the entries of scoring vectors that follow the $\ell$-th $k$-pattern in the initial region $R_0$. These equality constraints are included in all (sub)regions that are considered by the algorithm. Hence, the regions that will be contained in the pool ${\cal P}$ when \texttt{mRegions} terminates will all satisfy the equality restrictions. In this way, the scoring vector that will be computed will follow the $\ell$-th $k$-pattern, as desired.

Our algorithm \texttt{ApxPSR}$_k$ first calls \texttt{mRegions} to compute the $\lceil d/k \rceil$ best scoring rules that follow the $\ell$-th $k$-pattern for $\ell=1, ..., \lceil d/k \rceil$ and returns the best among all these rules, i.e., the one that yields the highest gain among them. Algorithm \texttt{ApxPSR}$_k$ follows as Algorithm~\ref{alg:apxpsr}.

\begin{algorithm}[h!]
    \SetAlgoLined
    \KwIn{parameter $d$, profile $\boldsymbol{\delta}(\Pi)$, set $C$ of constraints}
    \KwOut{a scoring vector $\s = (s_1, ..., s_d)$}
	$S:=\emptyset$\\
	\For{$\ell \in [\lceil d/k \rceil]$}{
		$S:=S\cup \mbox{\texttt{mRegions}}(d,\boldsymbol{\delta}(\Pi),C,k,\ell)$
	 }   
	 return $\s \in \arg\max_{\s'\in S}\{\g(\s',\boldsymbol{\delta}(\Pi),C) \}$
\caption{\texttt{ApxPSR}$_k$}
\label{alg:apxpsr}
\end{algorithm}

The next theorem summarizes the properties of algorithm \texttt{ApxPSR}$_k$. Observe that the algorithm runs in polynomial time when the parameter $k$ is a constant. The approximation ratio is better than approximately $k$ times that of \texttt{BestApproval}.

\begin{theorem}\label{thm:apxpsr}
Given an instance of {\sf OptPSR} with parameter $d$, a set of constraints $C$, and a profile $\Pi$, algorithm \texttt{ApxPSR}$_k$ runs in time {\em $\bigO(|C|^k \cdot \poly(|C|,d))$} and returns a $\lceil d/k\rceil^{-1}$-approximate solution. 
\end{theorem}

\begin{proof}
We first show the bound on the running time. Observe that \texttt{ApxPSR}$_k$ selects the best scoring vector among those returned in $\lceil d/k\rceil$ executions of \texttt{mRegions}. We will show that each execution of \texttt{mRegions} takes time $\bigO(|C|^k\cdot \poly(|C|,d))$.

By mimicking the proof of Theorem~\ref{thm:improved-alg}, we can view the expansion of the pool by \texttt{mRegions} as a non-complete binary tree $T$ with nodes corresponding to regions. Then, the total time required to find all regions by \texttt{mRegions} is proportional to the size of $T$, which is at most its height $|C|$ times the number of leaves. The number of leaves is again the number of different non-empty regions, which is upper-bounded by the number of different sign patterns that the quantities $\boldsymbol{\delta}(x,y,\Pi)\cdot \s$ define for each constraint $(x,y)$ in $C$. Since these $|C|$ quantities are linear functions over $k$ (as opposed to $d$ in the proof of Theorem~\ref{thm:improved-alg}) coordinates of vector $\s$, the results of Alon \cite{A96} and Warren \cite{W68} yield that the total number of different sign patterns is at most $\left(\frac{8e|C|}{k}\right)^k$. For each of the nodes of $T$, feasibility can again be checked by solving two linear programs with $d$ variables and at most $|C|$ constraints in time $\poly(|C|,d)$. The bound on the running time follows. 

In order to prove the bound on the approximation ratio, we will show that there exists a scoring vector $\s^{(k)}$ that follows some $k$-pattern which is a $\left\lceil\frac{d}{k}\right\rceil^{-1}$-approximate solution to the {\sf OptPSR} instance. Then, the scoring vector $\s$ returned by algorithm \texttt{ApxPSR}$_k$ will have the same approximation guarantee, since \texttt{ApxPSR}$_k$ returns the best scoring vector following some $k$-pattern, i.e., $\g(\s,\boldsymbol{\delta}(\Pi),C) \geq \g(\s^{(k)},\boldsymbol{\delta}(\Pi),C)$.

Consider an instance of {\sf OptPSR} consisting of a profile $\Pi$ and a set of constraints $C$ with weighting $w:C\rightarrow \mathbb{R}_{\geq 0}$. Let $\s^* \in \mathbb{R}_{\geq 0}^d$ be the optimal scoring vector for this instance. We will show that there exists a scoring vector $\s^{(k)}$ that follows some $k$-pattern such that $\g(\s^{(k)},\boldsymbol{\delta}(\Pi),C) \geq \left\lceil\frac{d}{k}\right\rceil^{-1} \g(\s^*,\boldsymbol{\delta}(\Pi),C)$.

Let $X\subseteq C$ be the set of constraints that are satisfied by the scoring vector $\s^*$. We now define an alternative view of $\s^*$ by setting $\alpha_d=s_d$ and $\alpha_i=s^*_i-s^*_{i+1}$ for $i=1, ..., d-1$. I.e., instead of keeping the entries of the scoring vector $\s^*$, we use the vector $\boldsymbol{\alpha}=(\alpha_1, ..., \alpha_d)$ to represent the entry $s^*_d$ and the increase of $s^*_i$ compared to $s^*_{i+1}$ for $i=1, ..., d-1$. Hence, $s^*_j=\sum_{i=j}^d{\alpha_i}$ for $j=1, ..., d$. Using the definition of $\boldsymbol{\delta}(x,y,\Pi)$, for every $(x,y)\in X$, we have 
\begin{align*}
\boldsymbol{\delta}(x,y,\Pi) \cdot \s^* = \sum_{j=1}^d{\delta_j(x,y,\Pi)\cdot s^*_j}
= \sum_{j=1}^d{\delta_j(x,y,\Pi)\sum_{i=j}^d{\alpha_i}}
= \sum_{i=1}^d{\alpha_i\sum_{j=1}^{i}{\delta_j(x,y,\Pi)}}.
\end{align*}
Now, define 
\begin{align*}
\xi_\ell(x,y,\Pi)&= \sum_{i=(\ell-1)k+1}^{\min\{d,\ell k\}}{\alpha_i\sum_{j=1}^{i}{\delta_j(x,y,\Pi)}}
\end{align*}
for $\ell\in [\lceil d/k\rceil]$, and observe that the last two equalities imply that  
\begin{align}\label{eq:sum}
\boldsymbol{\delta}(x,y,\Pi)\cdot \s^* &=\sum_{\ell=1}^{\lceil d/k\rceil}{\xi_{\ell}(x,y,\Pi)}.
\end{align} 
For $\ell\in [\lceil d/k\rceil]$, define the set
\begin{align*}
X_\ell = \{ (x,y)\in X: \xi_\ell(x,y,\Pi)>0 \}
\end{align*} 
and observe that, for every constraint $(x,y)\in X$, (\ref{eq:sum}) and the fact that $\boldsymbol{\delta}(x,y,\Pi)\cdot \s^*>0$ imply that $\xi_{\ell}(x,y,\Pi)>0$ for some $\ell\in [\lceil d/k\rceil]$. This yields that 
\begin{align*}
\g(\s^*,\boldsymbol{\delta}(\Pi),X) &= \sum_{(x,y)\in X}{w(x,y)} \leq \sum_{\ell=1}^{\lceil d/k\rceil}{\sum_{(x,y)\in X_{\ell}}{w(x,y)}}.
\end{align*} 
We conclude that there exists $\ell^*\in [\lceil d/k\rceil]$ such that 
\begin{align}\label{eq:apx}
\sum_{(x,y)\in X_{\ell^*}}{w(x,y)} &\geq \lceil d/k\rceil^{-1} \g(\s^*,\boldsymbol{\delta}(\Pi),X)
= \lceil d/k\rceil^{-1} \g(\s^*,\boldsymbol{\delta}(\Pi),C).
\end{align}

In order to complete the proof, we will define a particular scoring vector $\s^{(k)}$ that follows a $k$-pattern and satisfies all constraints in $X_{\ell^*}$. We do so by first defining the difference vector $\boldsymbol{\alpha}^{(k)}=(\alpha^{(k)}_1, ..., \alpha^{(k)}_d)$ corresponding to $\s^{(k)}$. This is done as follows:
\begin{itemize}
\item If $\ell^*>1$, we set $\alpha^{(k)}_i=0$ for $i=1, ..., k(\ell^*-1)$.
\item If $\ell^*<\lceil d/k \rceil$, we set $\alpha^{(k)}_i=0$ for $i=k\ell^*+1, ..., d$.
\item We set $\alpha^{(k)}_i=\alpha_i$ for $i=k(\ell^*-1)+1, ..., \min\{d,k\ell^*\}$.
\end{itemize}
Now, define the scoring vector $\s^{(k)}$ as $s^{(k)}_j=\sum_{i=j}^d{\alpha^{(k)}_i}$ for $j=1, ..., d$. Then, for every $(x,y) \in X$,
\begin{align*}
\boldsymbol{\delta}(x,y,\Pi)\cdot \s^{(k)} &= \sum_{j=1}^d{\delta_j(x,y,\Pi)\cdot s^{(k)}_j} = \sum_{j=1}^d{\delta_j(x,y,\Pi)\sum_{i=j}^d{\alpha^{(k)}_i}}\\
&= \sum_{i=1}^d{\alpha^{(k)}_i\sum_{j=1}^i{\delta_j(x,y,\Pi)}} = \sum_{i=k(\ell^*-1)+1}^{\min\{d,k\ell^*\}}{\alpha^{(k)}_i\sum_{j=1}^i{\delta_j(x,y,\Pi)}}\\
&= \xi_{\ell^*}(x,y,\Pi).
\end{align*}
This equality, together with the definition of set $X_{\ell}$ yields that $\boldsymbol{\delta}(x,y,\Pi)\cdot \s^{(k)}>0$ for every $(x,y)\in X_{\ell^*}$, i.e., the scoring vector $\s^{(k)}$ satisfies all constraints in $X_{\ell^*}$. We obtain that
\begin{align*}
\g(\s^{(k)},\boldsymbol{\delta}(\Pi),C) &\geq \g(\s^{(k)},\boldsymbol{\delta}(\Pi),X_{\ell^*})=\sum_{(x,y)\in X_{\ell^*}}{w(x,y)} \geq \lceil d/k\rceil^{-1}\g(\s^*,\boldsymbol{\delta}(\Pi),C).
\end{align*}
The last inequality follows by (\ref{eq:apx}). This completes the proof of the approximation bound.
\end{proof}

\section{Hardness of approximation}\label{sec:hard}
We devote this section to proving that {\sf OptPSR} is not simply computationally hard, but also hard to approximate well. The proof of our next statement relies on an approximation-preserving reduction from a well-known inapproximable optimization problem.

\begin{theorem}\label{thm:apx-hard}
For every constant $\eta \in (0,1/24]$, {\sf OptPSR} is NP-hard to approximate within $23/24+\eta$.
\end{theorem}

\begin{proof}
We use a reduction from MAX-3LIN-2, the problem of 
maximizing the number of satisfied equations in an over-determined system of linear equations modulo $2$. An instance of MAX-3LIN-2 consists of $n$ binary variables $x_i \in \{0,1\}$ and $m$ equations of the forms $x_i \oplus x_j \oplus x_k = 0$ and $x_i \oplus x_j \oplus x_k = 1$, where $\oplus$ denotes addition modulo $2$. The objective is to find an assignment to the variables so that the number of satisfied equations is maximized. Below, we use the term $\alpha$-equation to refer to an equation of the form $x_i \oplus x_j \oplus x_k = \alpha$ (for $\alpha\in\{0,1\}$).

Given an instance of MAX-3LIN-2, our reduction constructs in polynomial-time an instance of {\sf OptPSR} that has a scoring vector that satisfies constraints of total weight $11m+L$ if and only if the MAX-3LIN-2 instance has an assignment satisfying $L$ equations. A famous result by H{\aa}stad~\cite{H01} states that, for any constant $\eta' \in (0,1/2)$, it is hard to distinguish in time polynomial in $n$ and $m$ whether a given instance of MAX-3LIN-2 has an assignment that satisfies at least $(1-\eta')m$ equations or any assignment satisfies at most $(1/2+\eta')m$ equations. As a consequence of our reduction, we obtain that it is hard to distinguish between instances of {\sf OptPSR} that have a scoring vector that satisfies constraints of total weight at least $(12-\eta')m$ and instances of {\sf OptPSR} in which the total weight of the constraints satisfied by any scoring vector is at most $(23/2+\eta')m$. An inapproximability bound of $23/24+\eta$ (for every constant $\eta\in (0,1/24]$) then follows by standard arguments. The rest of the proof consists of two main parts: the description of the reduction and the proof of correctness. 

\paragraph{Description of the reduction} Without loss of generality, we can assume that the scoring vectors $\s = (s_1,s_2,...,s_d)$, that we seek for, have $s_1=d$ and the remaining scores are defined in terms of $d-1$ variables $a_1, a_2, ..., a_{d-1} \geq 0$ as $s_{i+1} = s_i - a_i$ (or, consequently, $s_j=d-\sum_{k=1}^{j-1}{a_k}$) for $i=1, ..., d-1$ so that $\sum_{i=1}^{d-1}{a_i}\leq d$. Hence, a constraint $(y,z)$ requiring that the score of $y$ is higher than the score of $z$ can be expressed as a linear inequality of the variables $a_j$ with $j \in [d-1]$. The assumption that $s_1=d$ allows for inequalities that have non-zero constant terms. 

Our reduction will be described in two steps. Given the MAX-3LIN-2 instance (i.e., the binary variables and the equations), we will first define linear inequalities among the variables $a_i$. Later, we will define the {\sf OptPSR} instance by explicitly constructing the profile and specifying the constraints as pairs of alternatives and corresponding weights that are consistent to these linear inequalities. 

For the first step of the description of our reduction, we present the linear inequalities that will later evolve to constraints of the {\sf OptPSR} instance. We set $\epsilon$ to be a small constant such that $0<\epsilon\leq 1/d$ and $1/\epsilon$ is an integer. We define the following inequalities:
\begin{itemize}
\item For every variable $x_i$, we have the four inequalities $a_i > 0$, $a_i < \epsilon$, $a_i > 1$ and $a_i < 1+\epsilon$. An important property of the four inequalities corresponding to variable $x_i$ is that three of them are simultaneously satisfied when $a_i \in  (0,\epsilon) \cup (1,1+\epsilon)$, and only two of them are satisfied for any other value of $a_i$. Intuitively, by enforcing these inequalities to be satisfied as constraints with high weight will simulate binary assignments in the MAX-3LIN-2 instance.

\item For every equation, there are four inequalities.
\begin{itemize}
\item If the equation is of the form $x_i \oplus x_j \oplus x_k = 0$, the inequalities are $a_i + a_j + a_k > 0$, $a_i + a_j + a_k < \epsilon$, $a_i + a_j + a_k >2$ and $a_i + a_j + a_k < 2+\epsilon$. Now, three of the inequalities corresponding to a $0$-equation $x_i \oplus x_j \oplus x_k = 0$ are simultaneously satisfied when $a_i + a_j + a_k \in (0,\epsilon) \cup (2, 2+\epsilon)$; otherwise, exactly two inequalities are satisfied. 

\item If the equation is of the form $x_i \oplus x_j \oplus x_k = 1$, the inequalities are $a_i + a_j + a_k > 1$, $a_i + a_j + a_k < 1+\epsilon$, $a_i + a_j + a_k >3$ and $a_i + a_j + a_k < 3+\epsilon$. Again, for every $1$-equation $x_i \oplus x_j \oplus x_k = 1$, we have three inequalities that are simultaneously satisfied when $a_i + a_j + a_k \in (1,1+\epsilon) \cup (3, 3+\epsilon)$; otherwise, exactly two inequalities are satisfied. 
\end{itemize}
These inequalities will correspond to constraints of equal (and low) weight of the {\sf OptPSR} instance. Overall, our construction will guarantee that the total weight of all satisfied constraints will be linear to the number of satisfied equations in the MAX-3LIN-2 instance. This will be crucial in proving the correctness (and the approximation-preserving nature) of the reduction.
\end{itemize} 

We are now ready to proceed with the second step of the description of the reduction. In particular, we show how the above inequalities are implemented by explicitly constructing a profile and pairs of alternatives as constraints. Let $d=n+1$. Also, for $i\in [d-1]$, let $m_i$ be the number of equations in which variable $x_i$ participates. For every variable $x_i$, with $i\in [n]$, we have four constraints $(y_i^t,z_i^t)$, where $t \in \{1,2,3,4\}$, of weight $m_i$ each. In the profile, alternatives $y_i^t$ and $z_i^t$ appear in specific positions as follows:
\begin{itemize}
\item Alternative $y_i^1$ appears once in position $i$, and alternative $z_i^1$ appears once in position $i+1$.  Then, the constraint $(y_i^1, z_i^1)$ corresponds to the inequality $s_i-s_{i+1}>0$ or, equivalently, $a_i > 0$.
\item Alternative $y_i^2$ appears once in position $1$ and $d/\epsilon$ times in position $i+1$, and alternative $z_i^2$ appears $d/\epsilon$ times in position $i$. The constraint $(y_i^2, z_i^2)$ corresponds to the inequality $s_1 - \frac{d}{\epsilon}s_i + \frac{d}{\epsilon}s_{i+1} > 0$ or, equivalently, $a_i < \epsilon$ (since $s_1=d$ and $a_i = s_i - s_{i+1}$).
\item Alternative $y_i^3$ appears $d$ times in position $i$, and alternative $z_i^3$ appears once in position $1$ and $d$ times in position $i+1$. The constraint $(y_i^3, z_i^3)$ corresponds to the inequality $-s_1 + ds_i - ds_{i+1} > 0$ or, equivalently, $a_i > 1$.
\item Alternative $y_i^4$ appears $1/\epsilon + 1$ times in position $1$ and $d/\epsilon$ times in position $i+1$, and alternative $z_i^4$ appears $d/\epsilon$ times in position $i$. The constraint $(y_i^4, z_i^4)$ corresponds to the inequality $\left(\frac{1}{\epsilon}+1\right)s_1 -\frac{d}{\epsilon}s_i + \frac{d}{\epsilon}s_{i+1} > 0$ or, equivalently, $a_i < 1+\epsilon$.
\end{itemize}

For every equation $\ell \in [m]$, we have four constraints $(b_\ell^t,c_\ell^t)$, where $t \in \{1,2,3,4\}$, of unit weight each. In the profile, these alternatives appear in specific positions depending on whether equation $\ell$ is a $0$- or a $1$-equation. In the case where it is a $0$-equation of the form $x_i \oplus x_j \oplus x_k = 0$, we have:
\begin{itemize}
\item Alternative $b_\ell^1$ appears once in positions $i$, $j$ and $k$, and alternative $c_\ell^1$ appears once in positions $i+1$, $j+1$ and $k+1$. Then, the constraint $(b_\ell^1,c_\ell^1)$ corresponds to the inequality $s_i - s_{i+1} + s_j - s_{j+1} + s_k - s_{k+1} > 0$ or, equivalently, $a_i + a_j + a_k > 0$.

\item Alternative $b_\ell^2$ appears once in position $1$ and $d/\epsilon$ times in positions $i+1$, $j+1$ and $k+1$, and alternative $c_\ell^2$ appears $d/\epsilon$ times in positions $i$, $j$ and $k$. The constraint $(b_\ell^2,c_\ell^2)$ corresponds to the inequality $s_1 - \frac{d}{\epsilon}s_i + \frac{d}{\epsilon}s_{i+1} - \frac{d}{\epsilon}s_j + \frac{d}{\epsilon}s_{j+1} - \frac{d}{\epsilon}s_k + \frac{d}{\epsilon}s_{k+1} > 0$ or, equivalently, $a_i + a_j + a_k < \epsilon$ (since $s_1=d$, $a_i=s_i-s_{i+1}$, $a_j=s_j-s_{j+1}$ and $a_k=s_k-s_{k+1}$).

\item Alternative $b_\ell^3$ appears $d$ times in positions $i$, $j$ and $k$, and alternative $c_\ell^3$ appears $2$ times in position $1$ and $d$ times in positions $i+1$, $j+1$ and $k+1$. The constraint $(b_\ell^3,c_\ell^3)$ corresponds to the inequality 
$-2s_1 + ds_i - ds_{i+1}+ ds_j - ds_{j+1}+ ds_k - ds_{k+1} >0$ or, equivalently, $a_i + a_j + a_k > 2$.

\item Alternative $b_\ell^4$ appears $2/\epsilon+1$ times in position $1$ and $d/\epsilon$ times in positions $i+1$, $j+1$ and $k+1$, and alternative $c_\ell^4$ appears $d/\epsilon$ times in positions $i$, $j$ and $k$. The constraint $(b_\ell^4,c_\ell^4)$ corresponds to the inequality 
$\left( \frac{2}{\epsilon}+1 \right)s_1 - \frac{d}{\epsilon}s_i + \frac{d}{\epsilon}s_{i+1} - \frac{d}{\epsilon}s_j + \frac{d}{\epsilon}s_{j+1} - \frac{d}{\epsilon}s_k + \frac{d}{\epsilon}s_{k+1} > 0$ or, equivalently, $a_i + a_j + a_k < 2+\epsilon$.
\end{itemize}

In the case where equation $\ell$ is a $1$-equation of the form $x_i \oplus x_j \oplus x_k = 1$, we have:
\begin{itemize}
\item Alternative $b_\ell^1$ appears $d$ times in positions $i$, $j$ and $k$, and alternative $c_\ell^1$ appears once in position $1$ and $d$ times in positions $i+1$, $j+1$ and $k+1$. Then, the constraint $(b_\ell^1,c_\ell^1)$ corresponds to the inequality 
$-s_1 + ds_i - ds_{i+1} + ds_j - ds_{j+1} + ds_k - ds_{k+1} > 0$ or, equivalently, $a_i + a_j + a_k > 1$ (since $s_1=d$, $a_i=s_i-s_{i+1}$, $a_j=s_j-s_{j+1}$ and $a_k=s_k-s_{k+1}$).

\item Alternative $b_\ell^2$ appears $1/\epsilon+1$ times in position $1$ and $d/\epsilon$ times in positions $i+1$, $j+1$ and $k+1$, and alternative $c_\ell^2$ appears $d/\epsilon$ times in positions $i$, $j$ and $k$. The constraint $(b_\ell^2,c_\ell^2)$ corresponds to the inequality 
$\left( \frac{1}{\epsilon}+1 \right)s_1 - \frac{d}{\epsilon}s_i + \frac{d}{\epsilon}s_{i+1} - \frac{d}{\epsilon}s_j + \frac{d}{\epsilon}s_{j+1} - \frac{d}{\epsilon}s_k + \frac{d}{\epsilon}s_{k+1} > 0$ or, equivalently, $a_i + a_j + a_k < 1+\epsilon$.

\item Alternative $b_\ell^3$ appears $d$ times in positions $i$, $j$ and $k$, and alternative $c_\ell^3$ appears $3$ times in position $1$ and $d$ times in positions $i+1$, $j+1$ and $k+1$. The constraint $(b_\ell^3,c_\ell^3)$ corresponds to the inequality $-3s_1 + ds_i - ds_{i+1} + ds_j - ds_{j+1} + ds_k - ds_{k+1} > 0$ or, equivalently, $a_i + a_j + a_k > 3$.

\item Alternative $b_\ell^4$ appears $3/\epsilon+1$ times in position $1$ and $d/\epsilon$ times in positions $i+1$, $j+1$ and $k+1$, and alternative $c_\ell^4$ appears $d/\epsilon$ times in positions $i$, $j$ and $k$. The constraint $(b_\ell^4,c_\ell^4)$ corresponds to the inequality 
$\left( \frac{3}{\epsilon}+1 \right)s_1 -\frac{d}{\epsilon}s_i + \frac{d}{\epsilon}s_{i+1} - \frac{d}{\epsilon}s_j + \frac{d}{\epsilon}s_{j+1} - \frac{d}{\epsilon}s_k + \frac{d}{\epsilon}s_{k+1} > 0$ or, equivalently, $a_i + a_j + a_k < 3+\epsilon$.
\end{itemize}

In order for this profile to be valid, we use sufficiently many agents and additional alternatives (that do not appear in the constraints) as placeholders, so that the alternatives mentioned above have the appropriate number of appearances in the rankings. 

\paragraph{Proof of correctness}
We now prove that there exists a variable assignment for the MAX-3LIN-2 instance that satisfies $L$ of its equations if and only if there exists a scoring vector that satisfies constraints of total weight $11m+L$. As we have discussed above, this is enough to complete the proof.

Consider an assignment that satisfies $L$ of the equations. Consider the scoring vector defined by setting $a_i = x_i +\epsilon/4$ for $i \in [d-1]$. Recall that we use $\epsilon\leq 1/d$ and, hence, $\sum_{i=1}^{d-1}{a_i}<d$; this is sufficient so that the corresponding scoring vector $\s$ has non-negative entries. This scoring vector satisfies:
\begin{itemize}
\item three out of the four inequalities corresponding to any variable $x_i$, since $a_i = \epsilon/4 \in (0,\epsilon)$ when $x_i= 0$ and $a_i = 1+\epsilon/4 \in (1,1+\epsilon)$ when $x_i = 1$;
\item three out of the four inequalities corresponding to any satisfied $0$-equation $x_i \oplus x_j \oplus x_k = 0$ since $a_i + a_j + a_k = 3\epsilon/4 \in (0,\epsilon)$ when $x_i + x_j + x_k = 0$ and $a_i + a_j + a_k = 2+3\epsilon/4 \in (2,2+\epsilon)$ when $x_i + x_j + x_k = 2$;
\item two out of the four inequalities corresponding to any unsatisfied $0$-equation since $a_i +a_j + a_k \not\in (0,\epsilon)\cup(2,2+\epsilon)$ in that case (observe that $a_i+a_j+a_k=1+3\epsilon/4$ when $x_i+x_j+x_k=1$ and $a_i+a_j+a_k=3+3\epsilon/4$ when $x_i+x_j+x_k=3$);
\item three out of the four inequalities corresponding to any $1$-equation $x_i \oplus x_j \oplus x_k = 1$ since $a_i + a_j + a_k = 1+3\epsilon/4 \in (1,1+\epsilon)$ when $x_i + x_j + x_k = 1$ and $a_i + a_j + a_k = 3+3\epsilon/4 \in (3,3+\epsilon)$ when $x_i + x_j + x_k = 3$;
\item two out of the four inequalities corresponding to any unsatisfied $1$-equation since $a_i +a_j + a_k \not\in (1,1+\epsilon)\cup(3,3+\epsilon)$ then (again, observe that $a_i+a_j+a_k=3\epsilon/4$ when $x_i+x_j+x_k=0$ and $a_i+a_j+a_k=2+3\epsilon/4$ when $x_i+x_j+x_k=2$).
\end{itemize}
Hence, the total weight of the constraints satisfied is $3\sum_{i=1}^n{m_i}+3L+2(m-L) = 11m+L$, since $\sum_{i=1}^n{m_i} = 3m$ due to the fact that all equations have three variables and the sum $\sum_{i=1}^n{m_i}$ accounts for the total number of appearances of all variables in equations.

Conversely, assume that we are given a scoring vector that satisfies constraints of total weight $11m+L$; we will show that there exists an assignment to the variables of the MAX-3LIN-2 instance that satisfies $L$ equations. First, we show that we can transform the scoring vector into a (possibly) different one with $a_i=\epsilon/4$ or $a_i=1+\epsilon/4$ for $i\in [d-1]$, without decreasing the total weight of the satisfied constraints.

For a variable $a_i \not\in (0,\epsilon)\cup(1,1+\epsilon)$ we have that the satisfied inequalities are the following: exactly two out of the four variable inequalities and at most three out of the four inequalities for each of the $m_i$ equations in which the variable $x_i$ appears. This gives a weight of at most $2m_i+3m_i=5m_i$. By setting $a_i = \epsilon/4$, exactly three out of the four variable inequalities and at least two out of the four equation inequalities in which $x_i$ appears are satisfied, for a total weight of at least $5m_i$. Clearly, there is no loss in weight after this change in the value of $a_i$. So, in the following, we can assume that $a_i\in (0,\epsilon)\cup(1,1+\epsilon)$ for every variable $a_i$.

Now, we slightly modify the variable values as follows: for all variables $a_i \in (0,\epsilon)$ we set $a_i = \epsilon/4$ and for all variables $a_i \in (1,1+\epsilon)$ we set $a_i = 1+\epsilon/4$. The set of inequalities containing $a_i$ that were satisfied before the modification are still satisfied after the update as well. This is trivial for the variable inequalities. For $\beta \in \{0,1,2,3\}$, for an equation inequality of the form $a_i+a_j+a_k<\beta+\epsilon$ (respectively, $a_i+a_j+a_k>\beta$) that was satisfied before the modification, at most $\beta$ (respectively, at least $\beta$) of the three variables have values in $(1,1+\epsilon)$ before the modification. Clearly, the inequality is satisfied after the modification as well. 

So, we can assume that we have total weight of $11m+L$ from satisfied constraints with the variables $a_i$ taking values in $\{\epsilon/4, 1+\epsilon/4\}$. Hence, $3\sum_{i=1}^n{m_i}=9m$ comes as weight from satisfied variable inequalities (with three satisfied inequalities per variable). Then, the remaining weight comes from $2m+L$ satisfied equation inequalities. The definition of the reduction implies that there exist $L$ equations in the MAX-3LIN-2 instance so that three among the four corresponding inequalities are satisfied. Then, it is easy to inspect that, if three among the four equation inequalities are satisfied when variables take values in $\{\epsilon/4,1+\epsilon/4\}$, then the (binary) assignment $x_i=a_i-\epsilon/4$ satisfies their corresponding equation as well. This yields an assignment with (at least) $L$ satisfied equations and the proof is complete.
\end{proof}

\section{Experiments}\label{sec:exp}
In this section, we present our experiments, which should be viewed as complementary to our theoretical work in the previous sections. We report on the execution of algorithms on two real-world {\sf OptPSR} instances that we have generated as well as on numerous synthetic ones. In contrast to the theoretical analysis of the approximation algorithms in Section~\ref{sec:apx-approval} which focuses on worst-case instances, here we are mainly interested in the average-case behavior of algorithms or scoring rules in scenarios that are close to ones that are likely to appear in practice. This explains the findings described in the following, according to which the observed performance of simple algorithms/scoring rules is much closer to optimality compared, for example, to the performance guarantees in Theorem~\ref{thm:approval-upper}.

Our experimental setup involves two different scenarios, to which we refer to as {\em ppl} and {\em col}. Each scenario is defined by a set of alternatives and by a {\em profile template} (or, simply, a {\em template}). The alternatives in the ppl scenario are the $48$ countries that are listed in Table~\ref{tab:Populations} of \ref{sec:app}. For the col scenario, the alternatives are the $36$ cities that are listed in Table~\ref{tab:CoL}. The templates consist of equal-sized subsets of alternatives that each agent will be asked to rank. Specifically, the templates consist of $392$ sets of six alternatives each. The distribution of the alternatives to the different sets of the ppl and col templates is almost uniform; each country appears in at least $47$ and at most $52$ sets, while each city appears in at least $57$ and at most $70$ sets. The templates are used to produce profiles as follows. Each profile has exactly $392$ agents. Each agent is given a distinct set of alternatives from the template (see Figure~\ref{fig:input} for examples of such sets for ppl and col) and ranks these alternatives. The profile then consists of the rankings provided by all agents.

\begin{figure}

\centering
\begin{tabular}{|c|c|}
\hline
\ \ \ \ \  & Sydney, Australia \\\hline
\ \ \ \ \  & Oslo, Norway \\\hline
\ \ \ \ \  & Baghdad, Iraq \\\hline
\ \ \ \ \  & Vienna, Austria \\\hline
\ \ \ \ \  & Washington, USA \\\hline
\ \ \ \ \  & London, UK \\
\hline
\end{tabular}
\ \ \ \ \ \ \ \ \ \ \ 
\begin{tabular}{|c|c|}
\hline
\ \ \ \ \  & Greece \\\hline
\ \ \ \ \  & Switzerland \\\hline
\ \ \ \ \  & Nigeria \\\hline
\ \ \ \ \  & Thailand \\\hline
\ \ \ \ \  & China \\\hline
\ \ \ \ \  & Mexico \\
\hline
\end{tabular}

\caption{An example of the sets from the ppl and col templates given to some agent. The particular format was used for building the real-world profiles. The blank column at the left is where the corresponding participant was required to define her ranking by putting distinct numbers from $1$ to $6$.}
\label{fig:input}
\end{figure}

\paragraph{Input profiles} In our experiments, we used both real-world and synthetic data. Two real-world profiles (for the scenarios ppl and col, respectively) were collected as input from $392$ participants\footnote{Actually, we had prepared even more sets of alternatives that could be part of the template; as we did not manage to obtain more inputs, we restricted the templates and synthetic profiles to $392$ agents for comparison reasons.} in the PatrasIQ\footnote{\url{www.patrasiq.gr}} technology exhibition organized by our home institution in April 2016. Each participant was given distinct sets of six countries and six cities (see Figure~\ref{fig:input}) from the ppl and col templates, and was asked to rank the countries in terms of their population and the cities in terms of their cost of living. These two profiles are available at \texttt{preflib.org}~\cite{MW13,MW17} as dataset ED-00034.\footnote{It seems that together with the dataset ED-00025 that were collected and used by Mao et al.~\cite{MPC13}, these are among the very few existing voting profiles with an underlying ground truth ranking.}

Many different synthetic profiles (in each scenario) were obtained by simulating agents who rank alternatives randomly. Specifically, each agent provides a noisy estimate of the correct underlying ranking of the alternatives assigned to her, according to either the Bradley-Terry~\cite{BT52} or the Plackett-Luce~\cite{L59,P75} noise model. Both Bradley-Terry and Plackett-Luce are random utility models~\cite{SPX12,T27}. They are defined using an underlying (positive) utility $u_x$ associated with each alternative $x$ and assume that the correct outcome of the pairwise relation between two alternatives $x$ and $y$ depends on the comparison between the utilities $u_x$ and $u_y$ (the alternative with the highest utility is better). In particular, the utilities of the $48$ alternatives-countries in the ppl scenario are their populations according to information retrieved by \texttt{wikipedia.org} in April 2016 (see Table~\ref{tab:Populations}). Similarly, the utilities of the $36$ alternatives-cities in the col scenario are their cost of living indices as retrieved by \texttt{numbeo.org} during the same time period (see Table~\ref{tab:CoL}).

A Bradley-Terry (BT, in short) agent (implicitly) works as follows. She first decides relations between all pairs of alternatives in her set. For each pair of alternatives $x$ and $y$ with corresponding utilities $u_x$ and $u_y$, the agent decides to rank $x$ above $y$ with probability $\frac{u_x}{u_x+u_y}$ and $y$ above $x$ with probability $\frac{u_y}{u_x+u_y}$. If the relative ranks of all pairs of alternatives in her set (that have been computed separately) define a ranking, then this is the ranking provided by the agent. Otherwise, the whole process is repeated from scratch.

A Plackett-Luce (PL, in short) agent decides the ranking of the alternatives in her set $B$ sequentially. Starting from the first position, the next undetermined position in the ranking is filled by alternative $x \in B$ with probability $\frac{u_x}{\sum_{y \in B} u_y}$. After a random selection, the chosen alternative is removed from $B$ and the process continues for the next undetermined position and the remaining alternatives until all positions are filled.

\paragraph{Constraints} The constraints of the {\sf OptPSR} instances we consider in our experiments were defined as follows. In both scenarios, we have a constraint $(x,y)$ for each ordered pair of alternatives $x$ and $y$ such that $u_x>u_y$. For example, a constraint in the ppl scenario is the pair (China, Switzerland) as China is more populous than Switzerland. We consider three different weightings of the constraints, defining different {\sf OptPSR} instances. In particular, the weight of $(x,y)$ is either $1$, or equal to $u_x-u_y$, or equal to $\log{(u_x-u_y)}$. 

Unit weights are used when we care only about maximizing the number of correctly recovered pairwise comparisons between alternatives. However, there might be pairs that are really important to recover correctly, while some others are not. For example, in the ppl scenario, it might be important to conclude that China is ranked above Switzerland since their population difference is almost $1.3$ billion people. Using this reasoning, an error in the comparison between Cuba and Belgium (both with population around $11$ million) would not be that severe. Weighted and log-weighted (as opposed to unweighted) constraints have been introduced to capture this characteristic.

\paragraph{Evaluation} 

Since all profiles that we experimented with have $d=6$, one would expect that the exact algorithm \texttt{Regions} presented in Section~\ref{sec:constant} would be the obvious choice in order to come up with the optimal scoring rule. Unfortunately, for the size of {\sf OptPSR} instances that we considered (with ${48 \choose 2} =1128$ constraints for ppl and ${36 \choose 2} =630$ constraints for col), \texttt{Regions} (as well as algorithm \texttt{ApxPSR} from Section~\ref{subsec:apxpsr}) turned out to be really slow, even after implementing several heuristics that yield minor performance improvements. This rather disappointing outcome, together with the fact that $d$ is small, forced us to consider scoring vectors with discretized scores (e.g., which are multiples of $0.05$ or $0.02$) in order to come up with approximations of the optimal scoring vector. Similarly, we have implemented a simplified variant of \texttt{ApxPSR}$_2$ by searching over all vectors of the form $(1,s,0,0,0,0)$, $(1,1,1,s,0,0)$, and $(1,1,1,1,1,s)$ where $s$ is a multiple of $0.02$ between $0$ and $1$. For real-world profiles, this approach has yielded the vectors that are shown in Table~\ref{tab:vectors}.

We remark that these variants of \texttt{Regions} and \texttt{ApxPSR}$_2$ are the most time-consuming among the algorithms we have implemented. The total execution time of all experiments (which were conducted using Matlab R2017B) is approximately 44 hours using an Intel 12-core i7 desktop, with 32Gb of RAM, running Windows 7. This is amortized to approximately 10 seconds per instance for the variant of \texttt{Regions} and less than 1 second per instance for the variant of \texttt{ApxPSR}$_2$.

\begin{table}[htbp]
	\centering
	
	\begin{subtable}[t]{\textwidth}
		\centering
		\begin{tabular}{l c c}
			\noalign{\hrule height 1pt}\hline
			rule & ppl  & col \\\hline
			Optimal 		      & {$(1.00, 0.50, 0.35, 0.20, 0.15, 0.05)$}	    & {$(1.00, 0.90, 0.30, 0.30, 0.24, 0.00)$}	  \\
			\texttt{ApxPSR}$_2$ 		      & {$(1.00, 0.02, 0.00, 0.00, 0.00, 0.00)$}	    & {$(1.00, 1.00, 1.00, 0.34, 0.00, 0.00)$}	  \\
			\noalign{\hrule height 1pt}\hline
		\end{tabular}
		\caption{Unweighted constraints}
		\label{tab:unweighted-vectors}
	\end{subtable}
	
	\begin{subtable}[t]{\textwidth}
		\centering
		\begin{tabular}{l c c}
	\noalign{\hrule height 1pt}\hline
	rule & ppl  & col \\\hline
	Optimal 		      & {$(1.00, 0.65, 0.65, 0.35, 0.30, 0.25)$}	    & {$(1.00, 0.68, 0.68, 0.50, 0.22, 0.22)$}	  \\
	\texttt{ApxPSR}$_2$ 		      & {$(1.00, 0.30, 0.00, 0.00, 0.00, 0.00)$}	    & {$(1.00, 1.00, 1.00, 0.52, 0.00, 0.00)$}	  \\
	\noalign{\hrule height 1pt}\hline
\end{tabular}
	\caption{Weighted constraints}
		\label{tab:weighted-vectors}
	\end{subtable}
	
	\begin{subtable}[t]{\textwidth}
		\centering
		\begin{tabular}{l c c}
	\noalign{\hrule height 1pt}\hline
	rule & ppl  & col \\\hline
	Optimal 		      & {$(1.00, 0.60, 0.42, 0.20, 0.18, 0.08)$}	    & {$(1.00, 0.65, 0.65, 0.35, 0.30, 0.25)$}	  \\
	\texttt{ApxPSR}$_2$ 		      & {$(1.00, 0.02, 0.00, 0.00, 0.00, 0.00)$}	    & {$(1.00, 1.00, 1.00, 0.52, 0.00, 0.00)$}	  \\
	\noalign{\hrule height 1pt}\hline
\end{tabular}
\caption{Log-weighted constraints}
		\label{tab:log-weighted-vectors}
	\end{subtable}
	\caption{The positional scoring vectors returned by the variants of the optimal algorithm and algorithm 
		\texttt{ApxPSR}$_2$ on the real-world profiles. These vectors follow by searching a discretized spaces with scores that are multiples of either $0.05$ or $0.02$. }
	\label{tab:vectors}
\end{table} 

We compare the optimal {\sf OptPSR} solution (obtained as described above) to the scoring vector returned by algorithm \texttt{BestApproval} which was presented in Section~\ref{sec:apx-approval}, the variant of \texttt{ApxPSR}$_2$ (also implemented as described above) and to two well-known scoring rules: the \texttt{Borda} count scoring rule that is defined by the scoring vector $(5,4,3,2,1,0)$, as well as the \texttt{Harmonic} scoring rule (also known as Dowdall) which is defined by the vector $(1, 1/2,1/3, 1/4, 1/5, 1/6)$. Table~\ref{tab:avg-results} shows the performance of these algorithms/scoring rules in all {\sf OptPSR} instances that we experimented with. Each performance value is expressed as a percentage of the total weight of the constraints satisfied by an algorithm/scoring rule compared to the total weight of all constraints. The two columns labeled ``real data'' contain values that correspond to a single execution of an algorithm/scoring rule on the two real-world ppl and col profiles. The values in the remaining columns are averages over executions of algorithms/scoring rules on $1000$ (random) profiles with either BT or PL agents as well as their standard deviations. Table~\ref{tab:avg-results} is split in three parts to report the results for {\sf OptPSR} instances with unweighted (Table~\ref{tab:unweighted}), weighted (Table~\ref{tab:weighted}), and log-weighted constraints (Table~\ref{tab:log-weighted}).

\begin{table}[htbp]
\centering

\begin{subtable}[t]{\textwidth}
\centering
\begin{tabular}{l cc cccc cccc}
\noalign{\hrule height 1pt}\hline
& \multicolumn{2}{c}{real data}	& \multicolumn{4}{c}{synthetic (BT)} & \multicolumn{4}{c}{synthetic (PL)}	\\
& ppl  & col & \multicolumn{2}{c}{ppl} & \multicolumn{2}{c}{col} & \multicolumn{2}{c}{ppl} & \multicolumn{2}{c}{col}		\\
rule &  &  & avg & std & avg & std & avg & std & avg & std \\\hline
Optimal 		      & {81.83}	    & {83.97}	   & {94.54} 	    & 0.642   & {93.74}	     & 1.037	& {93.19} 	    & 0.916           & {88.20}	    & 1.867\\
\texttt{ApxPSR}$_2$ 		      & {80.50}	    & {82.22}	   & {93.22} 	    & 0.774   & {\bf 92.74}	     & 1.019	& {\bf 92.04} 	    & 0.963           & {\bf 86.68}	    & 1.941\\
\texttt{Borda}  & 79.96	    & 82.06 	       & {\bf 93.43} & 0.706   & {92.04} & 1.112	& {91.94} &0.983	        & {85.95} & 1.984\\
\texttt{Harmonic}         & {\bf 80.94} & {\bf 82.54}  & 92.85 		 & 0.746   & 91.35        & 1.297	   & 90.50	    & 1.109   & 83.18		     & 2.321	\\	
\texttt{BestApproval} & 79.43		        & 80.48	          & 91.72		 & 0.844   & 91.42		 & 1.080	   & 90.54     & 0.986   & 84.68		     & 1.929	\\
\noalign{\hrule height 1pt}\hline
\end{tabular}
\caption{Unweighted constraints}
\label{tab:unweighted}
\end{subtable}

\begin{subtable}[t]{\textwidth}
\centering
\begin{tabular}{l cc cccc cccc}
\noalign{\hrule height 1pt}\hline
& \multicolumn{2}{c}{real data}	& \multicolumn{4}{c}{synthetic (BT)} & \multicolumn{4}{c}{synthetic (PL)}	\\
& ppl  & col & \multicolumn{2}{c}{ppl} & \multicolumn{2}{c}{col} & \multicolumn{2}{c}{ppl} & \multicolumn{2}{c}{col}		\\
rule &  & & avg & std & avg & std & avg & std & avg & std \\\hline
Optimal 				& {95.98}	                 & {92.93}	 	   & {99.66} 	   & 0.078   & {98.69} 	  & 0.364  & {99.49}     & 0.134	   & {96.02}	  & 1.079\\
\texttt{ApxPSR}$_2$ 				& {\bf 95.62}	                 & {91.95}	 	   & {99.47} 	   & 0.117   & {\bf 98.34} 	  & 0.412  & {\bf 99.31}     & 0.172	   & {\bf 95.13}	  & 1.236\\
\texttt{Borda}	        & 94.67		         & 91.92		       & {\bf 99.52} & 0.102	& {98.07} & 0.447	& {99.30} & 0.166  & {94.82} & 1.260 \\
\texttt{Harmonic}		& {95.42} & {\bf 92.01}   & 99.42 	       & 0.124	& 97.80		  & 0.542 & 99.04		  & 0.229	& 92.67	      & 1.769\\	
\texttt{BestApproval}	& 95.24		 & 90.83		       & 99.34		       & 0.143	& 97.90		  & 0.472  	& 99.10		  & 0.196   & 94.18		  & 1.313
\\
\noalign{\hrule height 1pt}\hline
\end{tabular}
\caption{Weighted constraints}
\label{tab:weighted}
\end{subtable}

\begin{subtable}[t]{\textwidth}
\centering
\begin{tabular}{l cc cccc cccc}
\noalign{\hrule height 1pt}\hline
& \multicolumn{2}{c}{real data}	& \multicolumn{4}{c}{synthetic (BT)} & \multicolumn{4}{c}{synthetic (PL)}	\\
& ppl  & col & \multicolumn{2}{c}{ppl} & \multicolumn{2}{c}{col} & \multicolumn{2}{c}{ppl} & \multicolumn{2}{c}{col}		\\
rule &  &  & avg & std & avg & std & avg & std & avg & std \\\hline
Optimal					& {83.03}		& {88.21}			& {95.29}		& 0.553 & {97.02}   	& 0.717 & {94.06}		& 0.768 & {92.53} & 1.487\\
\texttt{ApxPSR}$_2$					& {81.56}		& {86.82}			& {94.07}		& 0.681 & {\bf 96.24}   	& 0.765 & {\bf 93.06}		& 0.887 & {\bf 91.24} & 1.695\\
\texttt{Borda}			& 81.05	 		& 86.91				& {\bf 94.31}	& 0.596 & {95.87}	& 0.841 & {92.94}	& 0.827  & {90.75} & 1.663 \\
\texttt{Harmonic}		& {\bf 82.15}	& {\bf 87.24}		& 93.75			& 0.674 & 95.32			& 0.952 & 91.59			& 1.015 & 88.08 & 2.076 \\
\texttt{BestApproval}	& 80.56	   		& 85.48   			& 92.79			& 0.722 & 95.52			& 0.841 & 91.65			& 0.872 & 89.77 & 1.633 \\
\noalign{\hrule height 1pt}\hline
\end{tabular}
\caption{Log-weighted constraints}
\label{tab:log-weighted}
\end{subtable}
\caption{Performance of algorithms/scoring rules on instances with weighted constraints. Each value is a percentage and denotes the ratio of the total weight of the satisfied constraints over the total weight of all constraints. For synthetic profiles with BT and PL agents, the values indicate average performance (avg) and standard deviations (std) from $1000$ simulations (on randomly generated BT and PL profiles), while for real-world data each value corresponds to a single execution of an algorithm/scoring rule. The performance of the best among the non-optimal scoring rules is marked in bold. As the ppl and col instances have 1128 and 630 constraints, a respective difference of $0.09$ and $0.16$ in performance in subtable (a) corresponds to a gain difference of one constraint.}
\label{tab:avg-results}
\end{table} 

Our results indicate that \texttt{Harmonic} is better than both \texttt{Borda} and \texttt{BestApproval} in {\sf OptPSR} instances with real-world profiles. \texttt{Harmonic} is also better than \texttt{ApxPSR}$_2$ in all real-world profiles besides ppl with weighted constraints. On the other hand, \texttt{ApxPSR}$_2$ is better the best algorithm for all synthetic profiles besides ppl with BT agents, where \texttt{Borda} is slightly better. Even though \texttt{BestApproval} is never the best among the four algorithms/scoring rules, it always provides competitive results. The performance values of algorithms/scoring rules on the real-world profiles are considerably inferior to those for synthetic profiles. This indicates that the BT and PL noise models are rather idealized and do not reflect accurately the behavior of the agents in our real-world inputs. Still, as they allowed for many executions, they made possible the assessment of the algorithms/scoring rules in terms of robustness, as we will see shortly.

Table~\ref{tab:weighted} indicates that the performance values for weighted constraints are considerably higher than those for unweighted ones. This is to be expected since the total weight of correctly recovered constraints improves significantly when heavy pairwise relations are satisfied. The performance values for log-weighted constraints seem to lie in-between.

The low standard deviation values in Table~\ref{tab:avg-results} for synthetic profiles indicate that the performance values measured for each algorithm/scoring rule in the $1000$ executions are always sharply concentrated around their average values. Indicatively, we demonstrate this by plotting a point for each execution, which has the performance value of the optimal scoring rule as $x$-coordinate and the performance value of \texttt{Borda} as the $y$-coordinate. The twelve $1000$-point clouds in Figure~\ref{fig:Borda} correspond to the three different constraint weightings and the profiles with BT and PL agents for the ppl and col scenarios. All points in these clouds lie below the red dashed diagonal as \texttt{Borda} is in general suboptimal. The clouds of points comparing the optimal scoring rule with the remaining algorithms/scoring rules (not reported here) have similar structure and, like Figure~\ref{fig:Borda}, suggest that the average values in Table~\ref{tab:avg-results} robustly characterize the performance of algorithms/scoring rules on all the synthetic instances that we have considered. Another interesting observation from these experiments is that the optimal scoring vectors corresponding to the different points of the clouds in Figure~\ref{fig:Borda} are, in general, very different, in spite of the fact that the optimal performance values are so concentrated.


\begin{figure*}[p]
\centering

\begin{subfigure}{0.3\textwidth}
   \includegraphics[scale=0.36]{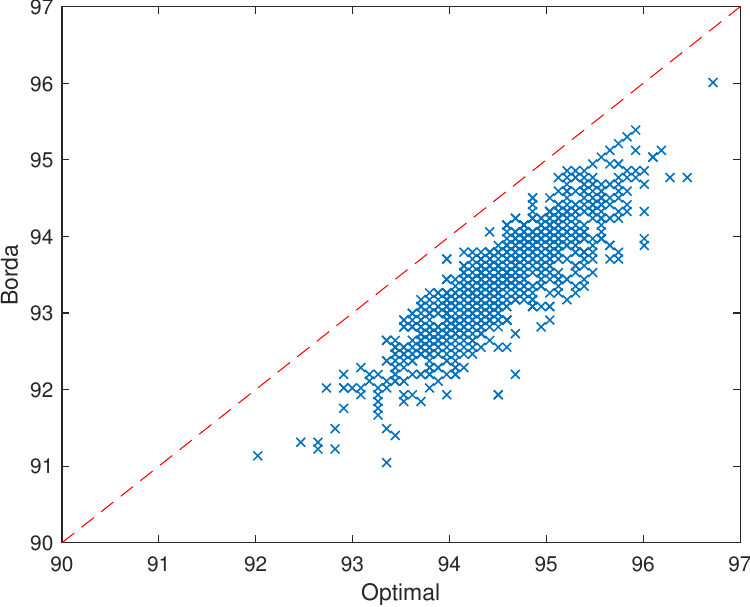} 
   \caption{BT, ppl, unweighted}
\end{subfigure}
\begin{subfigure}{0.3\textwidth}
   \includegraphics[scale=0.36]{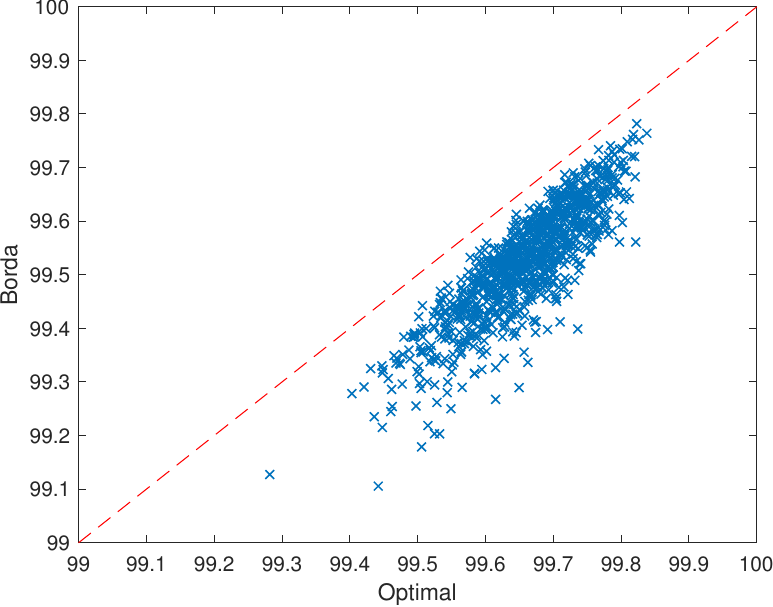} 
   \caption{BT, ppl, weighted}
\end{subfigure}
\begin{subfigure}{0.3\textwidth}
   \includegraphics[scale=0.36]{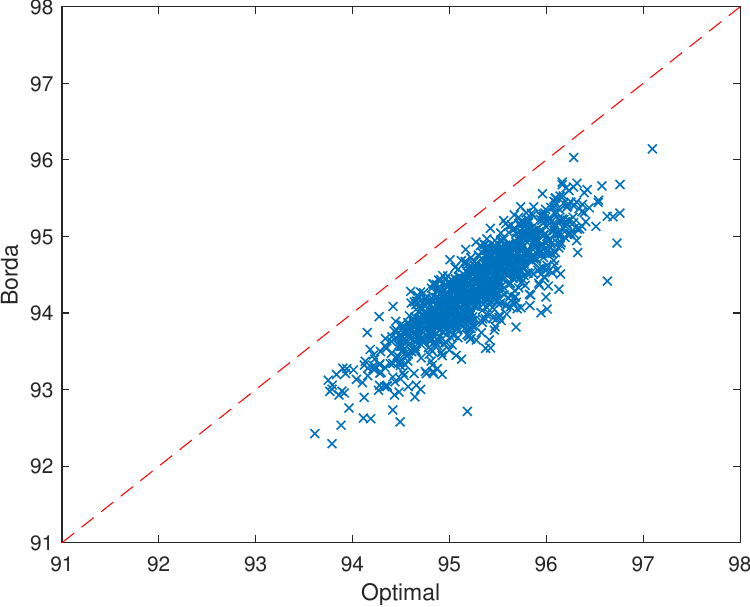} 
   \caption{BT, ppl, log-weighted}
\end{subfigure}
\vspace{10pt}

\begin{subfigure}{0.3\textwidth}
   \includegraphics[scale=0.36]{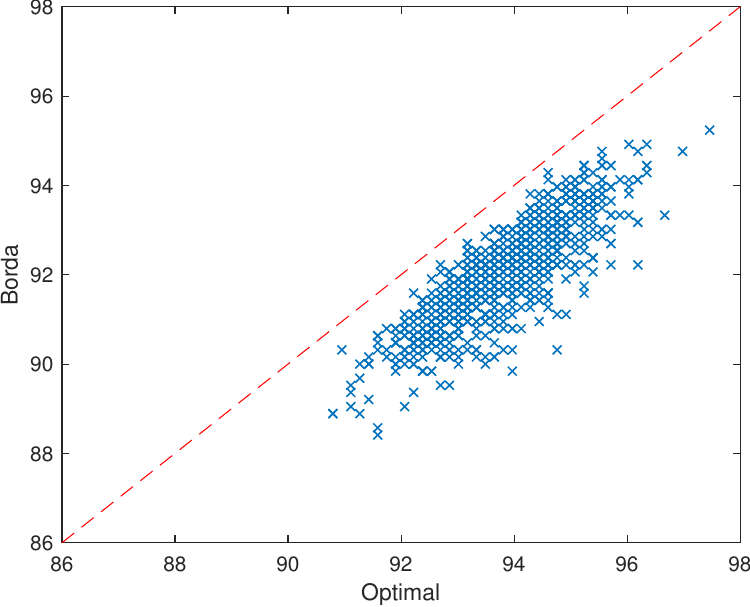} 
   \caption{BT, col, unweighted}
\end{subfigure}
\begin{subfigure}{0.3\textwidth}
   \includegraphics[scale=0.36]{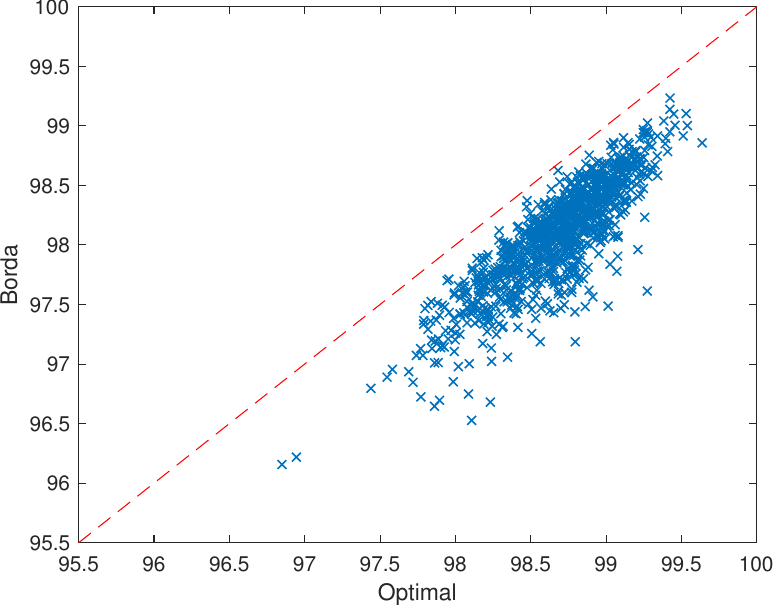} 
   \caption{BT, col, weighted}
\end{subfigure}
\begin{subfigure}{0.3\textwidth}
   \includegraphics[scale=0.36]{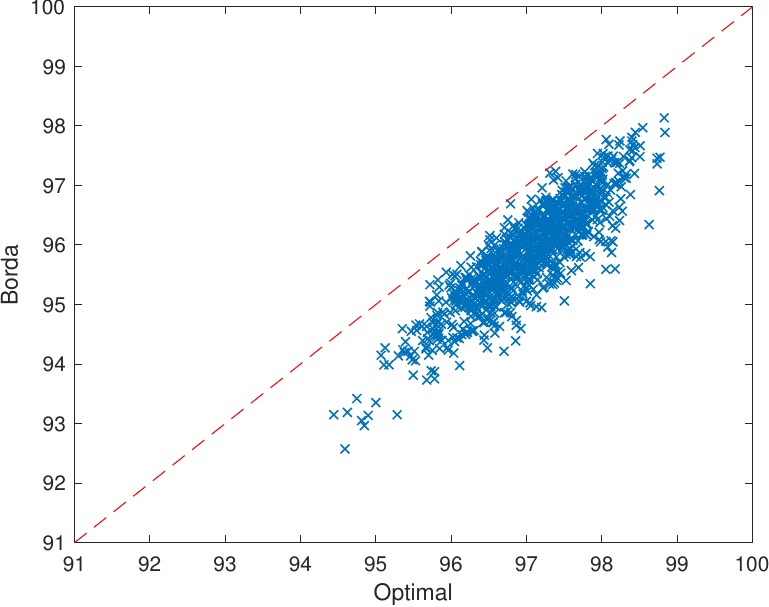} 
   \caption{BT, col, log-weighted}
\end{subfigure}
\vspace{10pt}

\begin{subfigure}{0.3\textwidth}
   \includegraphics[scale=0.36]{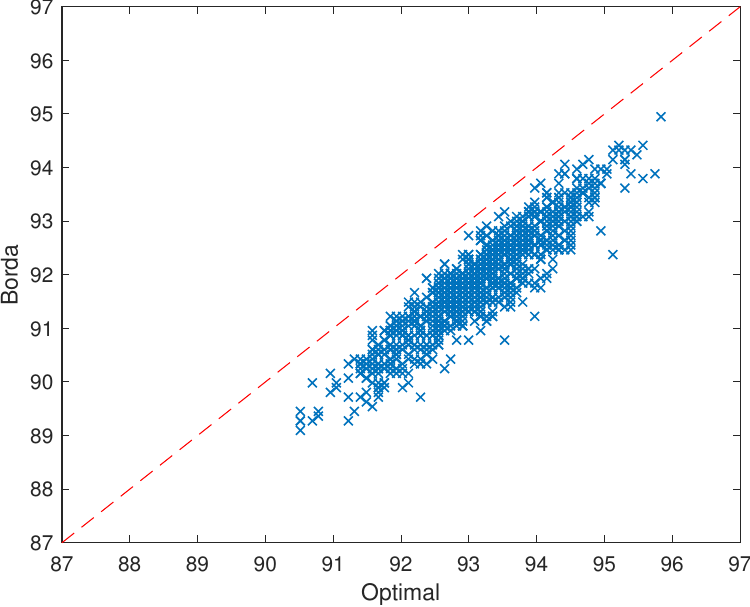} 
   \caption{PL, ppl, unweighted}
\end{subfigure}
\begin{subfigure}{0.3\textwidth}
   \includegraphics[scale=0.36]{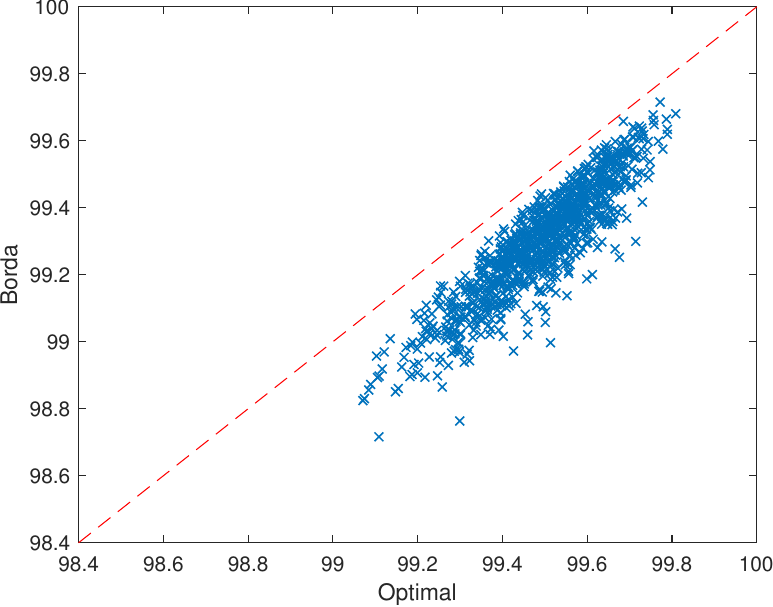} 
   \caption{PL, ppl, weighted}
\end{subfigure}
\begin{subfigure}{0.3\textwidth}
   \includegraphics[scale=0.36]{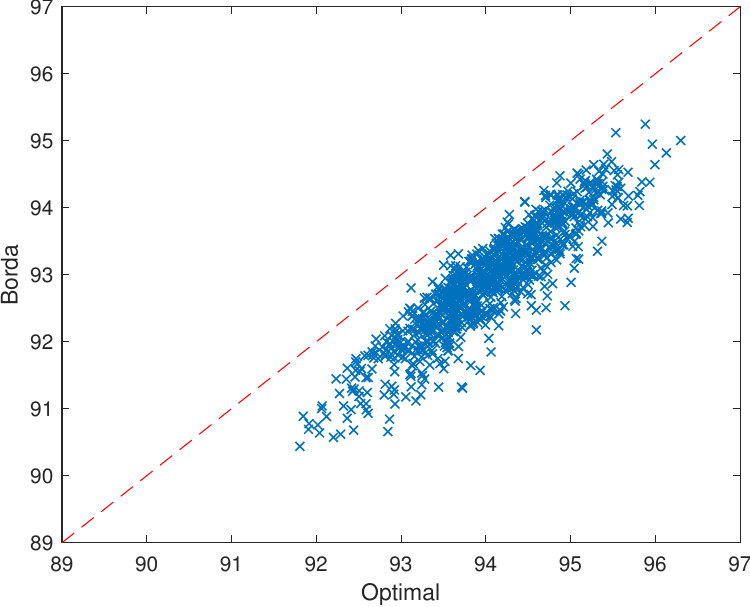} 
   \caption{PL, ppl, log-weighted}
\end{subfigure}
\vspace{10pt}

\begin{subfigure}{0.3\textwidth}
   \includegraphics[scale=0.36]{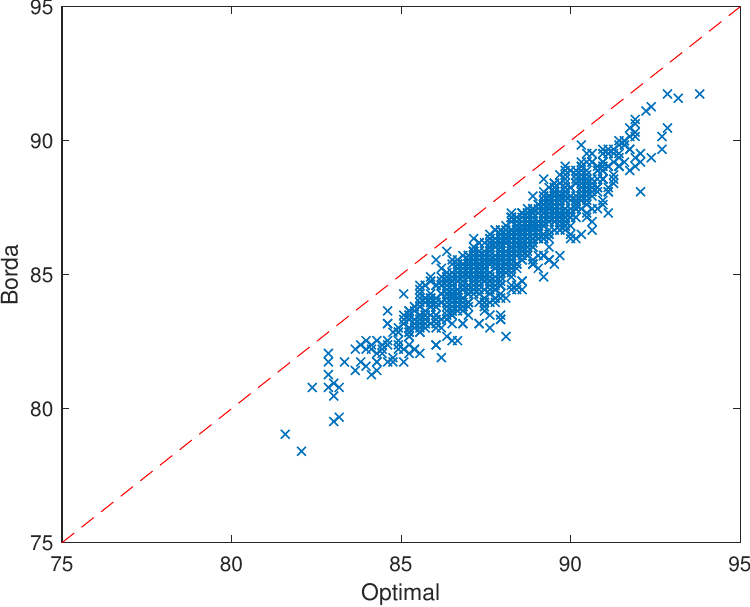} 
   \caption{PL, col, unweighted}
\end{subfigure}
\begin{subfigure}{0.3\textwidth}
   \includegraphics[scale=0.36]{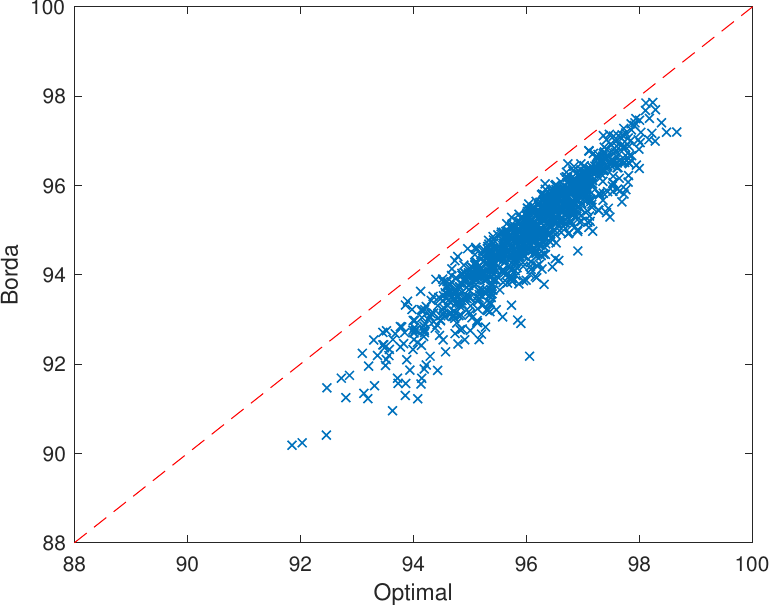} 
   \caption{PL, col, weighted}
\end{subfigure}
\begin{subfigure}{0.3\textwidth}
   \includegraphics[scale=0.36]{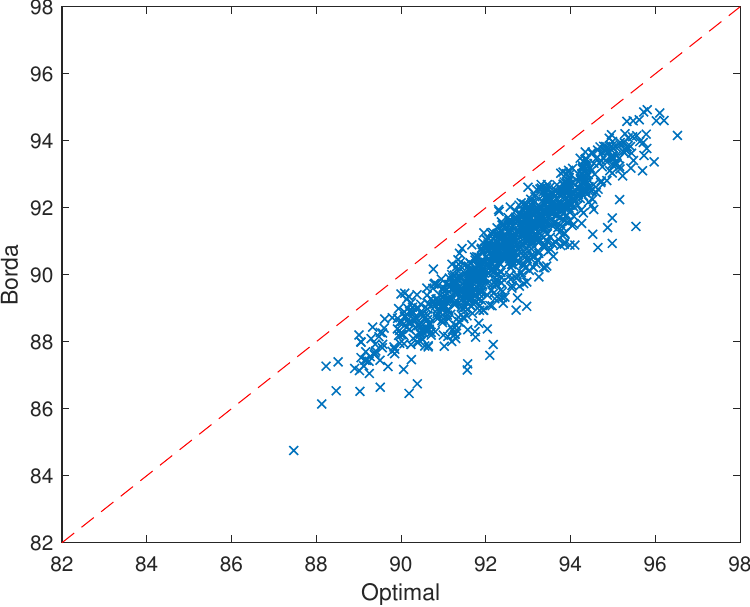} 
   \caption{PL, col, log-weighted}
\end{subfigure}

\caption{Comparing \texttt{Borda} with the optimal {\sf OptPSR} solution. Each cloud consists of $1000$ points, each corresponding to a distinct simulation. The $x$-coordinate of each point is the performance value of the optimal scoring rule and the $y$-coordinate is the performance value of \texttt{Borda}. The caption of each subfigure indicated the type of agents (BT or PL), the scenario (ppl or col), and the constraint weighting (unweighted, weighted, or log-weighted).}

\label{fig:Borda}
\end{figure*}


\section{Conclusions and open problems}\label{sec:open}
Motivated by crowdsourcing and rating applications, we have introduced and studied the {\sf OptPSR} problem. Very informally, the problem is to compute a positional scoring rule, whose outcome when applied on a given profile is as close as possible to constraints corresponding to underlying correct pairwise relations between alternatives. We have presented the algorithm \texttt{Regions} that solves the problem exactly by cleverly searching the space of all candidate scoring vectors and exploiting linear programming. The algorithm runs in polynomial time when the parameter $d$ (representing the number of alternatives in the ranking of each agent) is constant. We also consider approximation algorithms for {\sf OptPSR}. A simple algorithm, called \texttt{BestApproval}, that selects among all approval vectors the one that satisfies constraints of the highest weight is shown to achieve a tight approximation ratio of $1/d$. Our more sophisticated algorithm \texttt{ApxPSR} can achieve better approximation ratios at the expense of higher running times. We complement these positive results by showing that {\sf OptPSR} is a hard-to-approximate optimization problem. We also present an experimental evaluation of algorithms and scoring rules on real-world and synthetic instances.

Our work reveals several open problems. First, we would like to design exact algorithms that are practical. Our ambitious goal here is to be able to solve ---in reasonable time--- {\sf OptPSR} instances like the ones we used in our experiments (i.e., with $d$ up to $10$, approximately $50$ alternatives, and $1000$ constraints). 

Second, we would like to determine the approximability of {\sf OptPSR}. Currently, there is a huge gap between our positive algorithmic results in Theorems~\ref{thm:approval-upper} and~\ref{thm:apxpsr} and the inapproximability bound from Theorem~\ref{thm:apx-hard}; the former have a dependence on $d$ while the latter is a constant close to $1$. Is there a polynomial time algorithm with constant approximation ratio? Is there a sub-constant inapproximability bound? These questions are very important from the theoretical point of view. More importantly, we would like to design practical approximation algorithms that will be effective on huge {\sf OptPSR} instances. Here, we need both simplicity and efficiency; achieving these two goals simultaneously seems elusive at this point.

Third, observe that our theoretical results in Section~\ref{sec:apx-approval} focus on worst-case approximation guarantees. In addition to such studies, we would like to conduct theoretical analysis in random profiles that have been produced by Plackett-Luce or Bradley-Terry agents or, more importantly, by more realistic agents who follow appropriate generalizations of random utility models (see~\cite{SPX12}). In particular, the following optimization problem that is inspired by the flavor of our experiments is very appealing: Given a template, constraints, and statistical information (e.g., a noise model) describing the behavior of agents, compute the best algorithm or scoring rule that maximizes the expected total weight of satisfied constraints. Here, the expectation is taken over random {\sf OptPSR} instances with agents following the given noise model that are asked to rank the sets of alternatives in the template.

Finally, our definition of {\sf OptPSR} assumes that all agents rank the same number of alternatives. This feature has been used for proof-of-concept purposes here but, admittedly, it could be very restrictive in many applications. Extending {\sf OptPSR} by allowing different numbers of alternatives per agent (and more general definitions of scoring vectors) is important.  

Thinking beyond {\sf OptPSR}, one could consider optimization problems of similar flavor by replacing positional scoring rules by a class of voting rules defined over incomplete votes which can be identified by a small number of parameters (in the same way in which positional scoring rules are identified by the position scores). One possibility might be to consider the class of Kemeny-like voting rules which given a profile of possibly incomplete votes computes a full ranking of the alternatives that has the minimum possible total distance from the votes of the profile. Each voting rule in this class is identified by a distance function between rankings. Kemeny is such a voting rule for the Kendall-tau distance function (see~\cite{Z16}). Then, a natural optimization problem would aim for optimizing the distance function parameters so that the resulting voting rule, when applied on the given profile, returns a ranking that is as close as possible to the constraints of an underlying true ranking. This direction might be worth studying, taking extra care of computational issues (e.g., Kemeny is computationally hard to resolve) that do not arise in {\sf OptPSR}.

\bibliographystyle{plain}
\bibliography{incomplete}

\appendix

\section{Alternatives and utilities used in experiments}\label{sec:app}
Here we present the alternatives and the utilities that define the correct underlying ranking, the constraints, and the corresponding weightings in our experiments. Table~\ref{tab:Populations} contains the list of $48$ countries and their populations sorted in descending order, as retrieved from {\tt wikipedia} in April 2016 (ppl scenario). Table~\ref{tab:CoL} contains the list of $36$ cities and their corresponding cost of living index sorted in decreasing order, as retrieved from the website {\tt numbeo.com} in April 2016 (col scenario).

\begin{table}[h!]
\centering
\begin{tabular}{l r l r l r}
\noalign{\hrule height 1pt}\hline
countries 	& 	population 		& countries 	& 	population 	& countries 	& 	population \\\hline
China		&	1,375,880,000	& Iran			&	79,149,100	& Peru			&	31,488,700	\\
India		& 	1,287,180,000	& Turkey		&	78,741,053	& Australia		&	24,051,600	\\
USA			&	323,225,000	    & Thailand		&	65,273,832	& Romania		&	19,861,000	\\
Indonesia	&	258,705,000	    & Great Britain	&	65,097,000	& Chile			&	18,191,900	\\
Brazil		&	205,900,000	    & France		&	64,543,000	& Netherlands	&	17,003,600	\\
Pakistan	&	193,295,802	    & Italy			&	60,676,361	& Belgium		&	11,312,444	\\
Nigeria		&	186,988,000		& South Korea	&	51,569,536	& Cuba			&	11,238,317	\\
Bangladesh	&	160,197,000		& Colombia		&	48,608,000	& Greece		&	10,864,979	\\
Russia		&	146,544,710		& Kenya			&	47,251,000	& Czech Republic	&	10,553,843	\\
Japan		&	126,920,000		& Spain			&	46,423,064	& Portugal		&	10,374,822	\\
Mexico		&	122,273,500		& Argentina		&	43,590,400	& Sweden		&	9,866,670	\\
Philippines	&	103,083,100		& Ukraine		&	42,738,070	& Hungary		&	9,849,000	\\
Ethiopia	&	92,206,005		& Algeria		&	40,400,000	& Austria		&	8,699,730	\\
Vietnam		&	91,700,000		& Iraq			&	36,575,000	& Israel		&	8,489,400	\\
Egypt		&	90,755,700		& Canada		&	36,048,521	& Switzerland	&	8,306,200	\\
Germany		&	81,459,000		& Saudi Arabia	&	32,248,200	& Bulgaria		&	7,202,198	\\
\noalign{\hrule height 1pt}\hline\end{tabular}
\caption{The $48$ countries that are used as alternatives in the ppl scenario, ordered by population. Retrieved from {\tt wikipedia.org} (April 2016).}
\label{tab:Populations}
\end{table} 

\begin{table}[h!]
\centering
\begin{tabular}{l r  l r  l r}
\noalign{\hrule height 1pt}\hline
cities 		    &   col index  	&  cities 		&   col index & cities    		&   col index \\\hline
San Francisco	&	111.67		& Doha			&	68.99	& Barcelona			&	47.91	\\
Zurich			& 	106.19		& Stockholm		&	66.46	& Montreal			&	46.87	\\
New York		&	100.00		& Melbourne		&	62.25	& Lagos				&	44.19	\\
Lausanne		&	93.72		& Tel Aviv		&	61.89	& Nicosia			&	41.58	\\
London			&	89.23		& Munich		&	59.72	& Athens			&	37.19	\\
Washington		&	85.67		& Rome			&	58.93	& Istanbul    		&	34.74	\\
Boston			&	80.86		& Brussels		&	58.29	& Baghdad			&	33.42	\\
Oslo			&	76.31		& Toronto		&	55.75	& Patras			&	32.21	\\
Sydney			&	75.92		& Maastricht	&	54.70	& Budapest			&	30.69	\\
Tokyo			&	74.35		& Vienna		&	52.59	& City of Mexico	&	29.20	\\
Dubai			&	73.07		& Genoa			&	51.11	& Bucharest			&	27.10	\\
Copenhagen		&	71.12		& Berlin		&	48.96	& Mumbai			&	24.64	\\
\noalign{\hrule height 1pt}\hline
\end{tabular}
\caption{The $36$ cities that are used as alternatives in the col scenario, ordered by cost of living (plus rent) index. Retrieved from {\tt numbeo.com} (April 2016).}
\label{tab:CoL}
\end{table} 

\end{document}